\documentclass[letterpaper]{article} 
\usepackage[preprint]{aaai2027}  
\usepackage[hyphens]{url}  
\usepackage{graphicx} 
\usepackage{natbib}  
\usepackage{caption} 
\title{Temporal Routing in Static Networks: The~Schedule Completion Problem}
\author{
    Michelle Döring,
    Niklas Mohrin,
    George Skretas
}
\affiliations{
    Hasso Plattner Institute, University of Potsdam, Germany \\
    michelle.doering@hpi.de, niklas.mohrin@student.hpi.de, georgios.skretas@hpi.de
}

\usepackage[dvipsnames,table]{xcolor}

\newcommand*{\textcite}{\citet}
\newcommand*{\Textcite}{\Citet}
\newcommand*{\thmcite}{\citet} 

\usepackage{amsmath,amssymb}
\usepackage{csquotes} 
\usepackage{microtype}

\usepackage[shortlabels]{enumitem}

\usepackage{mathtools}
\usepackage{dsfont}        

\usepackage{url} 

\usepackage{graphicx} 

\usepackage{array}     
\usepackage{booktabs}  

\usepackage{subcaption}

\usepackage
[
    ruled,         
    vlined,        
    linesnumbered, 
    nokwfunc,      
]
{algorithm2e} 
\DontPrintSemicolon
\SetProcNameSty{textsc}
\newcommand*{\SetupAlgoInputOutput}{
    \SetKwInOut{AlgoInput}{Input}
    \SetKwInOut{AlgoOutput}{Output}
    \SetKwInOut{AlgoOutputx}{Output\hspace{0.1ex}}
}
\newcommand*{\SetupAlgoPrePost}{
    \SetKwInOut{AlgoPrec}{Precondition}
    \SetKwInOut{AlgoPost}{Postcondition}
}
\SetKw{Continue}{continue}

\usepackage{xspace}   
\usepackage{xparse}
\usepackage[textsize=scriptsize]{todonotes}

\usepackage[small,langfont=caps]{complexity}
\let\R\undefined

\let\G\undefined
\let\E\undefined
\let\SE\undefined

\usepackage{bracealign}
\usepackage[capitalise,noabbrev]{cleveref}

\let\originalleft\left
\let\originalright\right
\renewcommand{\left}{\mathopen{}\mathclose\bgroup\originalleft}
\renewcommand{\right}{\aftergroup\egroup\originalright}

\xspaceaddexceptions{]\}}

\allowdisplaybreaks

\usepackage{amsthm,thmtools}
\declaretheorem[numberwithin=section]{theorem}
\declaretheorem[numberlike=theorem]{corollary, lemma, proposition, conjecture, observation}
\declaretheorem[numberwithin=theorem]{claim}
\declaretheorem[style=definition, numberlike=theorem]{definition}

\declaretheorem{result}

\crefname{observation}{Observation}{Observations}
\Crefname{observation}{Observation}{Observations}

\newcommand*{\N}{\mathds{N}}

\newcommand*{\R}{\mathds{R}}

\DeclarePairedDelimiter\cardinality{\lvert}{\rvert}

\providecommand\with{}
\newcommand\SetSymbol[1][]{\nonscript\,#1\vert \allowbreak \nonscript\,\mathopen{}}
\DeclarePairedDelimiterX\set[1]{\lbrace}{\rbrace}{ \renewcommand\with{\SetSymbol[\delimsize]} #1 }
\DeclarePairedDelimiterX\multiset[1]{\llbracket}{\rrbracket}{ \renewcommand\with{\SetSymbol[\delimsize]} #1 }
\DeclarePairedDelimiterX{\norm}[1]{\lVert}{\rVert}{#1}

\newcommand\varcal[1]{\text{\usefont{OMS}{cmsy}{m}{n}#1}}
\renewcommand{\O}{\varcal{O}}

\makeatletter
\def\moverlay{\mathpalette\mov@rlay}
\def\mov@rlay#1#2{\leavevmode\vtop{%
        \baselineskip\z@skip \lineskiplimit-\maxdimen
        \ialign{\hfil$\m@th#1##$\hfil\cr#2\crcr}}}
\newcommand{\charfusion}[3][\mathord]{
    #1{\ifx#1\mathop\vphantom{#2}\fi
            \mathpalette\mov@rlay{#2\cr#3}
        }
    \ifx#1\mathop\expandafter\displaylimits\fi}
\makeatother

\newclass{\paraNP}{paraNP}

\newcommand*{\OPT}{\mathrm{OPT}}

\newcommand*{\G}{\mathcal G}
\newcommand*{\E}{\mathcal E}
\newcommand*{\lifetime}{\Lambda}
\newcommand*{\lifespan}{\tau}
\newcommand*{\len}{\ell}
\newcommand*{\SE}{\mathrm{SE}}
\newcommand*{\SEflow}{\widehat{\SE}}
\newcommand*{\SEcompressed}[1]{\SEflow_{#1}}
\newcommand*{\PS}{\mathrm{PS}}
\newcommand*{\IO}{\mathrm{IO}}
\newcommand*{\St}{\mathrm{St}}
\newcommand*{\gapsize}{\gamma}
\newcommand*{\vin}{\mathrm{in}}
\newcommand*{\vout}{\mathrm{out}}
\newcommand*{\seg}{\mathrm{seg}}

\newcommand*{\firstt}{\alpha}
\newcommand*{\tnext}{t_{\mathrm{next}}}
\newcommand*{\Concentrate}{\textsc{Concentrate}}
\newcommand*{\Sparsify}{\textsc{Sparsify}}
\newcommand*{\ShiftBackwards}{\textsc{ShiftBackwards}}
\newcommand*{\ShiftForwards}{\textsc{ShiftForwards}}

\newlang{\TEDSC}{TEDSC}
\newlang{\lenTEDSC}{$\ell$-TEDSC}
\newlang{\lifeTEDSC}{$\lifespan$-TEDSC}
\newlang{\MaxFlow}{Max Flow}
\newlang{\QuickestFlow}{Quickest Flow}
\newlang{\EarliestArrivalFlow}{Earliest Arrival Flow}
\newlang{\TSAT}{3-SAT}
\newlang{\BinPacking}{Bin Packing}
\newlang{\UnaryBinPacking}{Unary Bin Packing}
\newlang{\Partition}{Partition}
\newlang{\VDP}{VDP}
\newlang{\lenVDP}{$\ell$-VDP}
\newlang{\EDP}{EDP}
\newlang{\lenEDP}{$\ell$-EDP}
\newlang{\DSHP}{DSHP}
\newlang{\TEDW}{TEDW}
\newlang{\TVDW}{TVDW}
\newlang{\LinePlanning}{Line Planning}
\newlang{\TrainTimetabling}{Train Timetabling}
\newlang{\VehicleScheduling}{Vehicle Scheduling}

\newenvironment{tightcenter}
{\parskip=0pt\par\nopagebreak\centering}
{\par\noindent\ignorespacesafterend}

\usepackage{tikz}
\usetikzlibrary{calc}
\usepackage{framed}

\newlength{\RoundedBoxWidth}
\newsavebox{\GrayRoundedBox}
\newenvironment{GrayBox}[1]%
{\setlength{\RoundedBoxWidth}{\columnwidth-5ex}
    \def\boxheading{#1}
    \begin{lrbox}{\GrayRoundedBox}
        \begin{minipage}{\RoundedBoxWidth}%
            }{%
        \end{minipage}
    \end{lrbox}%
    \begin{tightcenter}%
        \begin{tikzpicture}%
            \node(Text)[draw=black!90,fill=white,rounded corners,%
                inner sep=2ex,text width=\RoundedBoxWidth]%
            {\usebox{\GrayRoundedBox}};
            \coordinate(x) at (current bounding box.north west);
            \node [draw=white,rectangle,inner sep=3pt,anchor=north west,fill=white]
            at ($(x)+(6pt,.7em)$) {\boxheading};
        \end{tikzpicture}
    \end{tightcenter}\vspace{0pt}%
    \ignorespacesafterend
}

\newenvironment{problem}[2][]{\noindent\ignorespaces%
    \FrameSep=8pt%
    \parindent=0pt%
    \vspace*{-.5em}
    \begin{GrayBox}{\textsc{#2}}%
        \newcommand\Prob{{Problem:}}%
        \newcommand\Input{{Input:}}%
        
        \begin{tabular*}{\columnwidth}{@{\hspace{.1em}} >{\itshape} p{1.05cm} p{0.8\columnwidth} @{\hspace{.1em}}}%
            }{
        \end{tabular*}%
    \end{GrayBox}%
    \vspace*{-.5em}
    \ignorespacesafterend
}

\begin{document}
\maketitle

\begin{abstract}
    We introduce the \textsc{Temporally Edge Disjoint Schedule Completion} ($\TEDSC$) problem in which we need to cover a set of temporal edge demands $D$ by routing $k$ temporal walks through a directed static graph while remaining temporally edge disjoint.
This problem combines the temporal aspects of train routing and passenger demands with the static nature of real-world rail networks.
We show how to solve $\TEDSC$ in polynomial time.
Motivated by real-world constraints, we next investigate two restricted variants of $\TEDSC$ in which each walk can travel only for some bounded distance or time~$h$.
For both variants, we present a  $(2-h^{-1})$-approximation algorithm and fully characterize the parameterized landscape with respect to $k$, $h$, and $\cardinality D$.
Surprisingly, if we restrict the underlying train network, the two variants diverge:
The distance variant stays $\W[1]$-hard parameterized by $k$ even on a path of three vertices, whereas the time variant admits a polynomial-time algorithm on every fixed bidirected star graph.

\end{abstract}

\section{Introduction}
\label{sec:intro}
Railway planning is one of the most extensively studied optimization problems in transportation~\cite{bast2016route}.
Three central stages of railway planning are line planning~\cite{schobel2012line}, timetabling~\cite{cacchiani2015tutorial}, and vehicle scheduling~\cite{bunte2009overview}, and they are traditionally solved sequentially~\cite{ceder2002urban}.
Based on passenger demands, the line planning problem selects a set of services (lines) and their operating frequencies.
Given these routes and frequencies, the timetabling problem generates exact arrival and departure times for them.
Finally, using this fixed timetable, the vehicle scheduling problem assigns vehicles to execute these trips.
Motivated by this planning hierarchy, we study a complementary temporal graph abstraction that changes the granularity of the optimization.
Classical railway planning treats complete trips as the elementary planning objects.
In contrast, we consider time-stamped edge traversals as elementary operational requirements and seek to assemble them into a minimum number of feasible temporal walks.
We formalize this optimization problem as follows.

We are given a directed static graph $G$ and a set of required traversals $D$, where each element $(u,v,t) \in D$ specifies that one of the temporal walks in the solution must traverse the edge $(u,v)\in E(G)$ at time step $t\in \N^+$.
We denote by $\G(G, D)$ the temporal graph where every edge of $G$ is available at all time steps until $\lifetime(D) = \max_{(u,v,t) \in D} t$.
The temporal walks should be \emph{temporally edge disjoint} (TED), that is, no two walks should use the same edge at the same time.
We call a set of TED temporal walks a \emph{schedule}.
The goal is to find a schedule $S$ such that for every $(u,v,t)\in D$ there is a walk in $S$ that traverses the temporal edge $(u,v,t)$.
Given $D$ and $k \in \N$, a schedule $S$ that covers all demands (abbreviated as $D \subseteq \E(S)$) and has size $\cardinality S \le k$ is called a \emph{feasible schedule}.
We refer to $D$ as the \emph{draft schedule}, since it specifies elementary traversals instead of complete routes.
The goal is to \emph{complete} this draft schedule by finding a feasible schedule.
We call this the \textsc{Temporally Edge Disjoint Schedule Completion} ($\TEDSC$) problem.
\Cref{fig:tedsc-example-instance} shows an example input and a feasible schedule.

\begin{problem}[]{Temporally Edge Disjoint Schedule Completion}
    \Input & Graph $G$, draft schedule $D$, integer $k$. \\
    \Prob  & Does there exist a schedule $S$ in $\G(G, D)$ with $D \subseteq \E(S)$ and $\cardinality{S} \le k$?
\end{problem}

\begin{figure}[h]
    \centering
    \includegraphics[width=0.9\columnwidth, page=1]{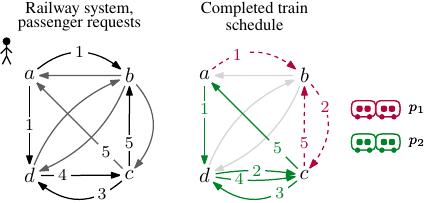}
    \caption{$\TEDSC$ instance with a feasible schedule of size $2$.}
    \label{fig:tedsc-example-instance}
\end{figure}

$\TEDSC$ captures the combinatorial problem of constructing a minimum number of vehicle routes satisfying prescribed time-stamped traversals.
In practice, however, vehicle operators' availability is often the binding constraint, as vehicle operators are subject to legal working-time regulations.
This motivates two constrained variants:
First, the length-constrained $\lenTEDSC$ problem, which adds the restriction that every walk in $S$ can contain at most $h$ temporal edges.
Second, the lifespan-constrained $\lifeTEDSC$ problem, where every walk in $S$ is allowed to be active for no more than $h$ time steps, meaning that the time difference between the start and end of a route cannot exceed $h$.
Intuitively, $\lenTEDSC$ bounds the amount of time that a vehicle travels on tracks, whereas $\lifeTEDSC$ bounds the time spent traveling plus the time spent waiting at stations.

The paper is organized as follows.
For the remainder of \cref{sec:intro}, we discuss our results and related work.
In \cref{sec:preliminaries}, we list preliminaries.
In \cref{sec:problem-def}, we formally define all studied variants of $\TEDSC$.
In \cref{sec:unconstrained,sec:hardness,sec:exact,sec:approx}, we discuss the techniques we use to obtain our results and describe the main ideas of our proofs.
In \cref{sec:conclusion}, we discuss future work.
The Technical Supplement of this submission contains an extended version of the paper, which includes proofs for all results and more detailed descriptions.

\subsection{Our Contribution}
\label{subsec:results}
We provide an algorithm that solves unconstrained $\TEDSC$ in polynomial time using a flow-based approach on the \emph{static expansion} of $\G(G, D)$.
As $\G(G, D)$ is defined implicitly and can span exponentially many time steps compared to the encoded size of the input instance, we perform a \emph{gap-compression} when constructing the static expansion.
\begin{result}
    $\TEDSC$ can be decided in polynomial time.
\end{result}

We move on to the constrained problems, $\lenTEDSC$ and $\lifeTEDSC$, starting with the parameterized complexity when parameterizing by $k$, $h$, and $\cardinality D$.
See \cref{tab:constrained_results} for an overview.
First, we show that both problems are $\NP$-hard, by providing a reduction from the $\TSAT$ problem.
The parameter $h$ is equal to $5$ in all constructed instances, which shows that routing trains for only short distances/durations is still difficult.
Additionally, we provide a parameterized reduction from the \lang{Edge Disjoint Paths} ($\EDP$) problem on planar DAGs.
By tracking the growth of the parameters in the reductions, we get two conditional lower bounds based on the Exponential Time Hypothesis (ETH).
\begin{result}
    Both $\lenTEDSC$ and $\lifeTEDSC$ are $\NP$-hard even if $h=5$.
\end{result}
\begin{result}
    Both $\lenTEDSC$ and $\lifeTEDSC$ parameterized by $k + \cardinality{D}$ are $\W[1]$-hard on planar DAGs.
\end{result}
\begin{result}
    Unless ETH fails, there is no algorithm that decides $\lenTEDSC$ or $\lifeTEDSC$ in time $2^{o(n + m + k + \cardinality D)}$ or in time $f(k + \cardinality D) n^{o(k + \cardinality D)}$.
\end{result}

On the positive side, we provide two algorithms that complement our hardness results.
Specifically, we construct an $\XP$ algorithm for constrained $\TEDSC$ parameterized by $k$, and an $\FPT$ algorithm for parameterizing by $k + h$.
\begin{result}
    Both $\lenTEDSC$ and $\lifeTEDSC$ parameterized by $k$ are in $\XP$.
\end{result}
\begin{result}
    Both $\lenTEDSC$ and $\lifeTEDSC$ parameterized by $k + h$ are in $\FPT$.
\end{result}

\begin{table}[t]
    \centering
    \begin{tabular}{ccc} \toprule
        Parameter            & $\lenTEDSC$                        & $\lifeTEDSC$ \\ \midrule
        $h$                  & \multicolumn{2}{c}{$\paraNP$-hard}                \\
        $k + \cardinality D$ & \multicolumn{2}{c}{$\W[1]$-hard}                  \\
        $k$                  & \multicolumn{2}{c}{$\XP$}                         \\
        $k + h$              & \multicolumn{2}{c}{$\FPT$}                        \\
        \bottomrule
    \end{tabular}
    \vspace{0.5em}
    \caption{
        The complexity of $\lenTEDSC$ and $\lifeTEDSC$ for the natural parameters $h$, $k$, and $\cardinality D$.
        Results for all other combinations of $k$, $h$, and $\cardinality D$ can be deduced using $k \le \cardinality D \le k \cdot h$.
    }
    \label{tab:constrained_results}
\end{table}

This completes the parameterized complexity landscape of the three problems for parameters $k$, $h$, and $\cardinality D$.
Next, we turn to solving constrained $\TEDSC$ on instances where the graph $G$ is restricted.
By reduction from $\BinPacking$, we show that $\lenTEDSC$ remains $\W[1]$-hard when parameterized by $k$, even if $G$ is a bidirected path of $3$ vertices.
To contrast this, we construct a polynomial time algorithm for $\lifeTEDSC$ on any fixed bidirected star graph.
\begin{result}
    $\lenTEDSC$ is $\W[1]$-hard when parameterized by $k$, even on the bidirected path graph with $3$ vertices.
\end{result}
\begin{result}
    There is an algorithm that solves $\lifeTEDSC$ on bidirected star graphs in $(n + \cardinality D)^{\O(n)}$ time.
\end{result}

Finally, we consider the optimization variant of $\TEDSC$ which asks for a schedule of minimum cardinality $\OPT$.
We provide an approximation algorithm using min-cost flows.
\begin{result}
    There is an algorithm that, given a \lang{Min-$\len$-TEDSC} or \lang{Min-$\lifespan$-TEDSC} instance, computes a feasible schedule $S$ with $\cardinality S \le (2 - h^{-1}) \OPT$ in polynomial time.
\end{result}

\subsection{Related work}
\label{sec:related-work}
$\TEDSC$ intersects three fundamental problems in transportation planning.
The $\LinePlanning$ problem models the selection of services that should be offered based on expected demands~\cite{schobel2012line,heinrich2022algorithms}.
The goal is to choose a number of routes through the transportation network and the frequencies at which those routes should be serviced.
This step does not yet include the times at which the lines should be serviced.
Instead, this is modeled by the $\TrainTimetabling$ problem~\cite{caprara2002modeling,cacchiani2015tutorial}.
Designing timetables corresponds to transforming predetermined walks in the static graph $G$ into temporal walks.
In this phase, the primary concerns are the logistics on the stations and tracks of the network.
Finally, once the routes and timetables are decided, the $\VehicleScheduling$ problem deals with allocating vehicles to the planned trips and accounts for tactical concerns such as fuel restrictions, vehicle maintenance, or delay tolerance~\cite{bertossi1987matching,carpaneto1989branch,lan2019optimizing}.

From a theoretic perspective, the problem of connecting pairs of terminal vertices using vertex or edge disjoint paths in static graphs is well studied.
The notions of vertex and edge disjointness are closely related as the two problems can be reduced to one another in polynomial time (see \cite{ahuja1993network} or \cref{sec:edp-reduction}).
In the special case that all terminal pairs connect the same two vertices $s$ and $t$, Menger's theorem~\cite{menger1927allgemeinen} tells us that the maximum number of disjoint paths is equal to the cardinality of the minimum cut.
This is the same as the value of a maximum flow from $s$ to $t$ if we assign unit capacities to all edges of the graph, and such a flow can be found in polynomial time~\cite{orlin2013max}.
In the general case, the endpoints of the paths are specified separately for each path, which drastically changes the complexity.
In the corresponding flow interpretation, each distinct source-sink pair constitutes a \emph{commodity} that demands one or more flow units to be routed from its source to its sink.
In general, the problem of deciding whether a multi-commodity flow of a given value exists is $\NP$-complete~\cite{karp1975computational}.
It remains $\NP$-hard even if there are only two commodities~\cite{even1976complexity}.
Their reduction from $\SAT$ uses one commodity with a single flow unit to choose the variable assignment and a second commodity with one flow unit for each clause to verify the assignment.
Recent work by \textcite{bodlaender2023parameterised} studies the complexity of two-commodity flow using structural graph parameters.
Returning to single-commodity flow, in addition to limiting how many flow units can traverse each edge, one can also introduce lower bounds for each edge that the sought flow must fulfill.
This makes the problem $\NP$-hard on undirected graphs~\cite{itai1978two}.
In directed networks, however, this requirement can be reduced to a standard flow problem~\cite{ford1962flows}.

The complexity of $\EDP$ parameterized by the number of paths $k$ greatly varies between directed and undirected graphs.
In undirected graphs, \textcite{robertson1995disjoint, kawarabayashi2012disjoint} give $\FPT$ algorithms, but in directed graphs, the problem is notoriously hard.
\Textcite{fortune1980directed} show $\NP$-hardness, even for $k = 2$.
For their reduction from $\SAT$, they design a gadget called the \emph{switch} which can be attached to two edges so that later on, only one of those two edges can be used.
The two paths leave the switch again and can be reused elsewhere in the graph, for example, to power additional switches.
Recently, \textcite{cavallaro2024edge} show that $\EDP$ parameterized by $k$ on Eulerian directed graphs is in $\FPT$.
\Textcite{chitnis2023tight} gives a tight lower bound for $\EDP$ on planar DAGs.

When considering paths of bounded length, Menger's theorem does not hold anymore~\cite{adamek1971remarks,exoo1983line}, which means that finding disjoint paths and cuts are separate problems.
\Textcite{itai1982complexity} study the problem of finding disjoint paths of bounded length.
They give an almost complete classification of the values of $h$ for which the problem of finding vertex or edge disjoint paths of length equal to or at most $h$ is $\NP$-hard.
In particular, they show that all variants are $\NP$-hard starting with $h = 5$.
\Textcite{golovach2011paths} study bounded length disjoint paths and cuts and give a complete characterization for the parameters $k$ and $h$.

The study of temporal paths begins with work by \textcite{cooke1966shortest} on shortest temporal walks with varying travel times.
\Textcite{halpern1974shortest} study shortest paths in temporal graphs and introduce what we today call the static expansion of a temporal graph.
The term \emph{temporal network} is first used by \textcite{kempe2002connectivity}.
\Textcite{klobas2023interference} study the problem of finding temporally vertex disjoint walks connecting the given terminal pairs.
They show that this problem is $\W[1]$-hard when parameterized by $k$.
Notably, in their problem, each walk must connect a predetermined pair of vertices, whereas in our problem, it is not fixed which demand is covered by which walk.
Recent work by \textcite{deligkas2025many} studies a similar problem to the one posed here:
Given a temporal graph, can it be exactly covered using $k$ temporal walks?
Thus, in their work, the temporal network is given in its entirety and the walks cannot move on edges that are not demanded in the input.
In the case of strict travel in directed networks (which is the closest to our setting), the authors give a polynomial time algorithm.
Finally, \textcite{bhaskar2025approximability} introduce a train routing problem in which a set of trains must be routed along temporal $s$-$t$-paths while maintaining a safe distance, the \emph{headway}, between them.
The concept of headway generalizes temporal edge disjointness and also incorporates travel times, whereas in our problem, all edges can be traversed in unit time.
They study approximation algorithms for minimizing the makespan (the latest arrival time among all trains).
Another perspective on the $\TEDSC$ problem is as a realization problem: a problem in which a temporal graph should be constructed so that it satisfies some property.
In our case, the sought graph would be the union of the temporal walks.
While there is much ongoing work on temporal realization problems~\cite{erlebach2025recognizing, klobas2025temporal, meusel2025directed}, to the best of our knowledge, none of it covers the problem studied here.

The problem of finding \emph{flows over time} (also called dynamic flows)~\cite{fleischer2007quickest,skutella2009introduction} is tangentially related to the problems studied here.
Particular problems like $\QuickestFlow$ and $\EarliestArrivalFlow$ use similar ideas, for example, that of a time-expanded network, but optimize in the time domain instead of fixing it like the draft schedule of $\TEDSC$.
To the best of our knowledge, no equivalent problem has been studied in this area.

\section{Preliminaries}\label{sec:preliminaries}
Let $\N = \set{0, 1, 2, \dots}$ denote the natural numbers including zero and $\N^+ = \N \setminus \set 0$.
For $n, m \in \N$ we write $[n, m] = \set{n, n + 1, \dots, m - 1, m}$ and $[n] = [1, n]$.
For a graph $G$, we abbreviate $n = \cardinality{V(G)}$ and $m = \cardinality{E(G)}$.
Throughout this work, we assume that the graphs given in problem instances do not contain any isolated vertices, as they are irrelevant for all discussed problems.
Unless stated otherwise, all graphs in this work are directed.
We say that a directed graph $G$ is \emph{bidirected} if for all $(u, v) \in E(G)$ we also have $(v, u) \in E(G)$~\cite{mehlhorn2008algorithms}.
For a walk $p$ in $G$, we denote the set of vertices and edges visited as $V(p)$ and $E(p)$, respectively.
We say that two walks $p$ and $q$ are \emph{vertex disjoint} if $V(p) \cap V(q) = \emptyset$ and \emph{edge disjoint} if $E(p) \cap E(q) = \emptyset$.
Given a set of edges $F \subseteq E(G)$, $p$ and $q$ are \emph{edge disjoint on $F$} if $E(p) \cap E(q) \cap F = \emptyset$.
Let $p'$ and $q'$ denote the subwalks of $p$ and $q$ with the first and last vertex removed.
If $p'$ and $q'$ are vertex disjoint, then $p$ and $q$ are \emph{internally vertex disjoint}~\cite{diestel2017graph}.
We denote the \emph{distance} of two vertices $u, v \in V(G)$ as $d_G(u, v)$.

\subsection{Temporal Graphs}

A \emph{temporal graph} $\G$ is a static graph $G$ equipped with a \emph{labeling} $\lambda\colon E(G) \to 2^{\N^+} \setminus \set \emptyset$.
The interpretation of $\lambda$ is that an edge $e$ is available only at the time steps $\lambda(e)$.
For $t \in \N$ let $E_t(\G) = \set{e \in E(G) \with t \in \lambda(e)}$ be the set of edges available at time step $t$.
We denote by $\lifetime(\G) = \max\set{t \in \N \with E_t(\G) \ne \emptyset}$ the \emph{lifetime} of $\G$.

We refer to the combination of an edge $(u, v) \in E(G)$ and a time step $t \in \lambda((u, v))$ as a \emph{temporal edge} $(u, v, t)$.
We denote the set of all temporal edges of $\G$ as $\E(\G)$.
A \emph{temporal walk} is a sequence of pairs of vertices and time steps $(v_0, t_0), (v_1, t_1), \dots, (v_k, t_k)$ such that the sequence of time steps $(t_i)_{i = 0}^k$ is strictly increasing and for each $i \in [k]$ the vertices are connected by the corresponding temporal edge $((v_{i-1}, v_i), t_{i-1}) \in \E(\G)$.
The literature distinguishes between \emph{strict} and \emph{non-strict} walks:
The former is the definition given here, whereas the latter requires that the sequence of time steps is only non-decreasing, that is, multiple edges may be traversed at the same time step.
We consider only strict temporal walks.
We denote by $\len(p) = k$ the \emph{length} of $p$ and by $\lifespan(p) = t_k - t_0$ the \emph{lifespan} of $p$.
Additionally, $\firstt(p) = t_0$ denotes the time step of the first temporal edge of $p$.
Furthermore, we define the set of temporal edges of a walk $p$ as $\E(p)$.
For a set of temporal walks $S$, we denote $\E(S) = \bigcup_{p \in S} \E(p)$.
Two temporal walks $p$ and $q$ are \emph{temporally edge disjoint}~(TED) if $\E(p) \cap \E(q) = \emptyset$.
A set of temporal walks is TED if its elements are pairwise TED.

When studying a temporal graph $\G$, the use of its \emph{static expansion} is a standard technique~\cite{kempe2002connectivity,wu2014path,mertzios2019temporal}.
The idea is to reduce the question of reachability in $\G$ to reachability in a suitable static graph, which is a time-expanded version of $\G$.

\begin{definition}[Static expansion]
    \label{def:static-expansion}
    Given a temporal graph $\G = \langle G, \lambda \rangle$, we define its static expansion $\SE(\G)$ to be the static graph with vertices
    \begin{align*}
        V(\SE(\G))        & = \set{v_t \with v \in V(G), t \in [\lifetime(\G) + 1]},
        \intertext{and edges}
        E_{\mathrm{wait}} & = \set{(v_t, v_{t+1}) \with v \in V(G), t \in [\lifetime(\G)]}, \\
        E_{\mathrm{move}} & = \set{(u_t, v_{t+1}) \with (u, v, t) \in \E(\G)},              \\
        E(\SE(\G))        & = E_{\mathrm{wait}} \cup E_{\mathrm{move}}.
    \end{align*}
    We refer to the vertices $V_t = \set{v_t}_{v \in V(G)}$ of a time step $t$ as the $t$-th layer of $\SE(\G)$.
    The edges in $E_{\mathrm{wait}}$ and $E_{\mathrm{move}}$ are called \emph{waiting edges} and \emph{moving edges}, respectively.
\end{definition}
\Cref{fig:tedsc-example-se} shows an example static expansion for the temporal graph of the $\TEDSC$ instance from \cref{fig:tedsc-example-instance}.
\begin{figure}[h]
    \centering
    \includegraphics[height=6cm, page=3]{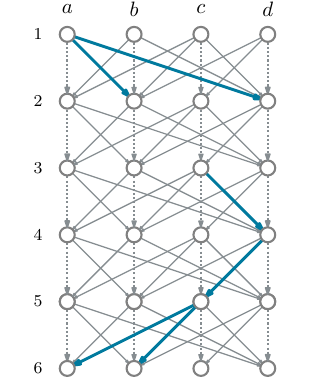}
    \caption{
        Example of a static expansion.
        The highlighted edges correspond to the temporal edges of the draft schedule.
    }
    \label{fig:tedsc-example-se}
\end{figure}
Observe that a temporal walk in $\G$ corresponds to a path in $\SE(\G)$ where taking a temporal edge is modeled by traversing a moving edge of $\SE(\G)$ and waiting on a vertex is modeled by traversing waiting edges.

\subsection{Network Flows}
\label{sec:flow-prelims}

A graph $G$ equipped with a capacity function $c\colon E(G) \to \N$ is called a \emph{network}.
Fix two distinct vertices $s, t \in V(G)$ called the \emph{source} and the \emph{sink} respectively.
An $s$-$t$-\emph{flow} is a function $f\colon E(G) \to \N$ satisfying $f(u, v) \le c(u, v)$ for all $(u, v) \in E(G)$ and $\sum_{(u, v) \in E(G)} f(u, v) = \sum_{(v, w) \in E(G)} f(v, w)$ for all $v \in V(G) \setminus \set{s, t}$.
We extend the functions $f$ and $c$ to all of $V(G) \times V(G)$ such that for $(u, v) \notin E(G)$ we have $f(u, v) = c(u, v) = 0$.
Without loss of generality, we assume that for all $u, v \in V(G)$, at least one of $f(u, v)$ and $f(v, u)$ is $0$.
We assume further that $f(\cdot, s) = f(t, \cdot) = 0$, that is, there is no flow entering the source $s$ and no flow leaving the sink $t$.
These assumptions can be enforced in polynomial time.
A flow $f$ can be decomposed into a (multi-)set of $s$-$t$-paths and simple cycles in $G$ such that $f(u, v)$ is equal to the number of paths and cycles that contain the edge $(u, v)$.
If $G$ is acyclic, $f$ decomposes into a (multi-)set of paths only.
The decomposition of a flow can be computed in $\O(n m)$ time, so we assume it is given whenever a flow is computed~\cite{ahuja1993network}.

The \emph{value} of $f$ is $\cardinality f = \sum_{v \in V(G)} f(s, v)$.
The problem of finding a flow of maximum value has been well studied.
It is possible to find such a flow in time $\O(n m)$ by using the algorithms by \textcite{king1994faster,orlin2013max}.

We can also solve extended variants of the $\MaxFlow$ problem.
The most relevant additions are \emph{lower bounds} and \emph{costs}.
In a network with lower bounds, every edge $e$ does not just specify a capacity but instead a range $r(e) = (l, u)$ so that a feasible flow must satisfy $l \le f(e) \le u$.
For directed networks, the problem of finding a flow adhering to lower bounds can be reduced to the regular $\MaxFlow$ problem~\cite{ford1962flows}.
In a network with costs, we are given an additional function that assigns to every edge $e \in E(G)$ the value~$\mathrm{cost}(e) \in \R$ and the objective is to find a feasible flow $f$ that minimizes $\mathrm{cost}(f) = \sum_{e \in E(G)} f(e) \mathrm{cost}(e)$.
Note that the empty flow has cost $0$, but might (a) not be feasible due to lower bounds and (b) not be optimal as costs can be negative.
If the network contains cycles with an overall negative cost, a min-cost flow does not exist.
In this work, all flow networks are DAGs and all costs are non-negative, so these problems do not affect us.
It is possible to find a minimum-cost flow $f$ with a fixed value in polynomial time~\cite{tardos1985strongly}.

\subsection{Satisfiability}

The $\SAT$ problem is central to the study of computational complexity.
It is the first problem which was shown to be $\NP$-hard~\cite{cook1971complexity,karp2009reducibility}.
An instance of $\SAT$ is a boolean formula $\varphi$ in conjunctive normal form, meaning:
\begin{enumerate}
    \item There are $n$ variables $\set{x_i}_{i = 1}^n$.
    \item The formula $\varphi$ is the conjunction of $m$ clauses $\set{C_j}_{j = 1}^m$, that is, $\varphi = C_1 \land C_2 \land \dots \land C_m$.
    \item Each clause is the disjunction of $k_i$ literals, that is, for all $j \in [m]$, $C_j = \ell_1 \lor \ell_2 \lor \dots \lor \ell_{k_i}$ where each $\ell_p = x_i$ or $\ell_p = \lnot x_i$ for some $i \in [n]$.
\end{enumerate}
The objective is to decide whether there is an assignment $A\colon [n] \to \set{0, 1}$ such that $\varphi$ is true when setting $x_i = A(i)$ for each $i \in [n]$.
The $k$-$\SAT$ problem is a special case of $\SAT$ where all $k_i \le k$.
Notably, for all $k \ge 3$, $k$-$\SAT$ is $\NP$-complete.

Assuming $\P \ne \NP$ means that there cannot be an algorithm that solves $\SAT$ in time $(n + m)^{\O(1)}$.
However, in the context of parameterized complexity, a stronger assumption is often used that gives rise to much stronger lower bounds: the \emph{Exponential Time Hypothesis} (ETH).
Intuitively, ETH states that there is no algorithm that solves $\TSAT$ in $2^{o(n)} (n + m)^{\O(1)}$ time.
\begin{conjecture}[ETH, \thmcite{impagliazzo2001complexity}]
    Let $\delta_3$ be the infimum of the set of all constants $c$ for which there exists an algorithm solving $\TSAT$ in time $2^{cn} \cdot \poly(n, m)$.
    It holds that $\delta_3 > 0$.
\end{conjecture}
By tracking the relations of parameters in parameterized reductions, we can prove lower bounds for other problems too, as we will do in this paper.
In particular, if the parameter $k$ of a parameterized problem is bounded as $\O(n)$, then this implies that under ETH, no algorithm can solve said problem in time $2^{o(k)}$.
While this already allows for strong lower bounds, it turns out that using the so-called \emph{sparsification lemma}, an even stronger lower bound for $\TSAT$ holds that allows us to also depend on the number of clauses.
\begin{proposition}[\thmcite{cygan2015parameterized}]
    \label{thm:eth-nm}
    Unless ETH fails, there is no algorithm that solves $\TSAT$ in $2^{o(n + m)}$ time.
\end{proposition}
For a more thorough introduction to the topic, we refer the reader to the textbook by \textcite{cygan2015parameterized}.
All statements in this work that assume ETH clearly state this assumption upfront.
Note that ETH implies $\FPT \ne \W[1]$~\cite{cygan2015parameterized}.

\section{The Schedule Completion Problem}
\label{sec:problem-def}
Consider a rail network modeled by a static directed graph $G$, where each edge represents a track that can be used by at most one train at any given time.
A \emph{draft schedule} is a set $D$ of demands, where each demand $(u, v, t)$ is a passenger request to travel along the edge $(u, v) \in E(G)$ at time step $t \in \N$.
We define $\lifetime(D) = \max_{(u,v,t) \in D} t$ as the \emph{lifetime} of the draft schedule; when $D$ is clear from context we write just $\lifetime$.
We want to operate a small fleet of trains, each following a route through the network, such that all passenger requests are covered.
Each train corresponds to a temporal walk, and since two trains cannot share the same track at the same time, the walks must be temporally edge disjoint (TED).
The \emph{schedule completion problem} asks: how many trains suffice?

Formally, since $G$ is a static graph with all edges always available, we model it as the temporal graph $\G(G, D) = \langle G, \lambda_D \rangle$ where $\lambda_D\colon e \mapsto [\lifetime(D)]$.
When $G$ and $D$ are clear from context, we abbreviate $\G = \G(G, D)$.
A \emph{schedule} for $D$ is a set $S$ of TED temporal walks in $\G(G, D)$.
In addition to the draft schedule $D$, we are given an integer $k \in \N$ that represents the number of available trains.
We say that $S$ is a \emph{feasible schedule} if $D \subseteq \E(S)$ (that is, for each demand $(u, v, t) \in D$ there is a walk $p \in S$ with $(u, v, t) \in \E(p)$) and $\cardinality S \le k$.
Without loss of generality, we assume there is a demand at time step $1$, as otherwise the times of all demands can be shifted down.

\begin{problem}[]{Temporally Edge Disjoint Schedule Completion}
    \Input & Graph $G$, draft schedule $D$, integer $k$. \\
    \Prob  & Does there exist a schedule $S$ in $\G(G, D)$ with $D \subseteq \E(S)$ and $\cardinality{S} \le k$?
\end{problem}

\subsection{Constrained variants}
$\TEDSC$ is a deliberate simplification of the real-world problem: trains can travel arbitrarily far and operate indefinitely.
In practice, both constraints are bounded: the fuel capacity limits the distance a train can travel, and working-time regulations limit the time a train can operate.
We study two variants of $\TEDSC$ that capture these restrictions, one for each constraint.
The \emph{length-constraint} bounds the number of edges a train may traverse.
The \emph{lifespan-constraint} instead bounds the total active time of a train, that is, the time between its first and last move.

\begin{problem}[]{$\chi$-$\TEDSC$ ($\chi \in \set{\len, \lifespan}$)}
    \Input & Graph $G$, draft schedule $D$, integers $k, h$. \\
    \Prob  & Does there exist a schedule $S$ in $\G(G, D)$ with $D \subseteq \E(S)$, $\cardinality S \le k$, and $\chi(p) \le h$ for all $p \in S$?
\end{problem}

\Cref{fig:tedsc-example-solutions} shows example solutions for all $\TEDSC$ variants for the instance of \cref{fig:tedsc-example-instance}.

\begin{figure}[h]
    \centering
    \includegraphics[width=0.9\columnwidth, page=2]{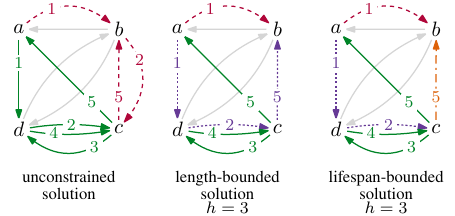}
    \caption{
        Three example solutions to (constrained) $\TEDSC$ instances:
        the unconstrained schedule consisting of $p_1=(a,d,c,d,a)$ and $p_2=(a,b,c,b)$;
        the length-constrained schedule $p_1=(a,d,c,b)$, $p_2=(c,d,c,a)$, and $p_3=(a,b)$;
        and the lifespan-constrained schedule $p_1=(a,d,c)$, $p_2=(c,d,c,a)$, $p_3=(a,b)$, and $p_4=(c,b)$.
    }
    \label{fig:tedsc-example-solutions}
\end{figure}

\subsection{Basic observations}
The two problems $\lenTEDSC$ and $\lifeTEDSC$ are closely related as the only difference between the length $\len$ and the lifespan $\lifespan$ of a walk is whether its remaining ``capacity'' decreases whenever it waits on a vertex.
As the lifespan of a temporal walk bounds its length, any schedule satisfying the lifespan constraint for some value $h$ also satisfies the length constraint for the same $h$.

\begin{observation}
    \label{thm:len-life-equiv}
    Let $\mathcal I$ be a $\lifeTEDSC$ instance and let $S$ be a feasible schedule.
    Then $S$ is also a feasible $\lenTEDSC$ schedule for $\mathcal I$.
\end{observation}

Recall that the goal of $\TEDSC$ is to cover multiple demands with each walk.
Our definition captures this by fixing the number of walks $k$ that may be used.
However, if an instance permits at least as many walks as demands, this instance can be solved by employing one walk for every demand.
\begin{observation}
    \label{thm:k-lt-D}
    Let $\mathcal I$ be a $\TEDSC$, $\lenTEDSC$, or $\lifeTEDSC$ instance.
    If $\cardinality D \le k$, then $\mathcal I$ is feasible.
\end{observation}

For the constrained variants, a similar bound on the number of demands holds for the other direction.
In particular, if there are few walks that can each travel only a small distance, then they can cover only few temporal edges.
We can thus immediately reject instances with particularly small $k$ and $h$.
\begin{observation}
    \label{thm:D-le-kh}
    Let $\mathcal I$ be a $\lenTEDSC$ or $\lifeTEDSC$ instance.
    If $k \cdot h < \cardinality D$, then $\mathcal I$ is infeasible.
\end{observation}
\begin{proof}
    For $\lenTEDSC$, suppose there is a feasible solution $S$.
    Every step by any walk covers at most one demand.
    Thus, the total length of all walks in $S$ is at least $\cardinality D$.
    As $\cardinality S \le k$,
    \[
        k \cdot h < \cardinality D \le \sum_{p \in S} \ell(p) \le k \cdot h,
    \]
    a contradiction.
    For $\lifeTEDSC$, the same proof works with the additional step of bounding $\len(p) \le \lifespan(p)$ in the final inequality.
\end{proof}

Our final observation holds only for $\lenTEDSC$.
Intuitively, if the length bound $h$ is large, it poses no restriction at all.
\begin{observation}
    \label{thm:length-h-le-nD}
    Let $\mathcal I = (G, D, k, h)$ be a $\lenTEDSC$ instance and suppose that $h \ge \cardinality{V(G)} \cdot \cardinality D$.
    Then $\mathcal I$ is feasible if and only if the corresponding unconstrained $\TEDSC$ instance $\mathcal I' = (G, D, k)$ is feasible.
\end{observation}
\begin{proof}
    The forward direction holds by definition.
    Suppose on the other hand that $\mathcal I'$ is feasible, then there is a schedule $S$ of TED walks fulfilling all demands.
    Let $p \in S$ and $p'$ be a contiguous subwalk of $p$ such that $p' \cap D = \emptyset$.
    We can assume that $p'$ is a path, because if $p'$ would visit a vertex $v$ more than once, then we could remove the cycle between the first and last visit of $v$ and still have a valid schedule.
    With this assumption, no walk can actually be longer than $n \cdot \cardinality D$ because each walk can cover at most $\cardinality D$ demands and the paths between them have length at most $n$.
    Therefore $S$ is also a schedule for $\mathcal I$.
\end{proof}
Based on these observations, we assume that $k < \cardinality D \le k \cdot h$ holds for all instances, and that $h < n \cdot \cardinality D$ holds for all $\lenTEDSC$ instances.

\section{Solving Unconstrained TEDSC}\label{sec:unconstrained}
    In this section, we solve unconstrained $\TEDSC$ by constructing the schedule completion flow network $(\SEflow(G,D), r)$ based on the static expansion $\SE(G,D)$ of $\G(G, D)$ and then computing a feasible flow (\cref{subsec:schedule-completion-flow}).
However, since the static expansion has $\O(n \cdot \lifetime(D))$ vertices and $\lifetime(D)$ can be exponential in the number of bits needed to encode the input, this yields only a pseudo-polynomial algorithm.

To achieve polynomial running time, we observe in \Cref{subsec:gaps-longtime} that large \emph{gaps} between consecutive demands can be compressed without affecting feasibility, bounding the flow network to $\poly(n, \cardinality{D})$ vertices.
This gap-observation will also play a central role for our subsequent results.

\subsection{Schedule Completion Flow Network} \label{subsec:schedule-completion-flow}

Recall from \cref{def:static-expansion} that the static expansion $\SE(G,D)=\SE(\G(G, D))$ is a static directed graph where every static path in $\SE(G,D)$ corresponds to a temporal walk in $\G$ and vice versa.
By adding a global source $s$ and sink $z$, we obtain an augmented graph $\SEflow(G,D)$ whose $s$-$z$-paths correspond to temporal walks in $\G$, where a feasible flow of value $k$ decomposes into $k$ paths, one per walk in the completed schedule.
The draft schedule and temporal edge disjointness is encoded as flow ranges: demand edges receive range $(1,1)$, forcing each demand to be covered by exactly one walk and preventing two walks from sharing a temporal edge; non-demand moving edges receive range $(0,1)$, enforcing temporal edge disjointness without requiring coverage.
\Cref{fig:tedsc-flow-network} shows the construction for an example instance.
\begin{figure*}[t]
    \centering
    \includegraphics[height=7.5cm]{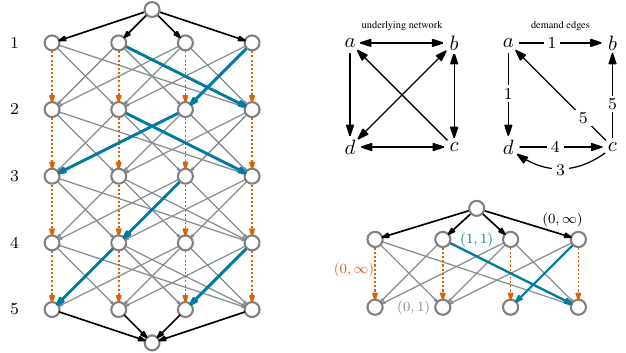}
    \caption{Illustration of the schedule completion flow network. The left part shows the graph $\SEflow(G,D)$, obtained from the static expansion $\SE(G,D)$ by adding a global source $s$ and sink $z$. Its edges consist of demand edges (solid, blue), waiting edges (dashed, orange), filler edges (solid, gray), and the edges incident with $s$ and $z$ (solid, black). The flow ranges are indicated in the two layers on the bottom right and also by the edge colors: orange and black edges have range $(0, \infty)$, blue edges have range $(1, 1)$, and gray edges have range $(0, 1)$.}
    \label{fig:tedsc-flow-network}
\end{figure*}
\begin{definition}[Schedule completion flow network]
    \label{def:tedsc-flow-network}
    Let $(G,D,k)$ be an instance of \TEDSC\ with lifetime $\lifetime = \lifetime(D)$.
    Starting from the static expansion $\SE(G,D)$ of the temporal graph $\G = \G(G, D)$, we define the augmented graph $\SEflow(G,D)$ by introducing a global source $s$ and sink $z$ connected to the first and last layer of $\SE(G,D)$, respectively.
    Formally, let $V_\SE = V(\SE(G, D))$ and $E_\SE = E(\SE(G, D))$.
    We set
    \begin{align*}
        V(\SEflow(G,D)) & = V_\SE \cup \set{s, z},                                                   \\
        E(\SEflow(G,D)) & = E_\SE \cup \bigcup_{v \in V(G)} \set{(s, v_1),\ (v_{\lifetime + 1}, z)}.
    \end{align*}
    The \emph{schedule completion flow network} is the pair $(\SEflow(G,D), r)$, where the flow range $r$ for each edge is
    \begin{align*}
        r(u_t, v_{t + 1}) & = (1, 1),       &  & \text{if } (u, v, t) \in D                  \\
        r(u_t, v_{t + 1}) & = (0, 1),       &  & \text{if } (u, v, t) \in \E(\G) \setminus D \\
        r(\cdot, \cdot)   & = (0, \infty) , &  & \text{otherwise.}
    \end{align*}
\end{definition}
\begin{proposition}
    \label{thm:tec-poly-n-tau}
    $\TEDSC$ can be decided using one $\MaxFlow$ computation on a network with $\O(n \cdot \lifetime)$ vertices and $\O((n + m) \cdot \lifetime)$ edges.
\end{proposition}
\begin{proof}
    Let $(\SEflow(G,D), r)$ be the schedule completion flow network from \Cref{def:tedsc-flow-network}.
    The graph $\SEflow(G,D)$ has $\O(n \cdot \lifetime)$ vertices and $\O((n + m) \cdot \lifetime)$ edges.

    Since $G$ is directed and we consider temporal walks to be strict, $\SEflow(G,D)$ is acyclic.
    Therefore any feasible $s$-$z$-flow of value $k$ in $(\SEflow(G,D),r)$ decomposes into exactly $k$ $s$-$z$-paths (never revisiting any time-vertex $v_t\in V(G)\times[\lifetime]$), each of which corresponding to a temporal walk in $\G$ (possibly revisiting a vertex $v\in V(G)$).
    Furthermore, the flow ranges ensure that these walks cover all demands in $D$ and are temporally edge disjoint:
    Each demand edge $e$ has range $(1, 1)$, which enforces that the corresponding demand is covered by exactly one walk, and that no two walks use the same temporal edge, i.\,e., the walks are temporally edge disjoint. Each remaining moving edge has range $(0, 1)$, which enforces that no two walks share a temporal edge, but allows for the possibility that some temporal edges are not used by any walk.
    Lastly, the waiting edges and the edges from $s$ and to $z$ are unconstrained, allowing for any number of walks to wait in a vertex, to start at any vertex in the first layer, and to end at any vertex in the last layer.
    As a result, $\TEDSC$ is a YES-instance if and only if a feasible $s$-$z$-flow of value at most $k$ exists in $(\SEflow(G,D), r)$.
    Such a flow can be found in polynomial time with respect to the network size as discussed in \Cref{sec:flow-prelims}.
    Note however that this time is pseudo-polynomial: the bound is polynomial in $n$ and $\lifetime$, but $\lifetime$ can be exponential in the bit-length of the input.
\end{proof}

\subsection{Gaps and gap-compression} \label{subsec:gaps-longtime}
The algorithm from \Cref{thm:tec-poly-n-tau} is super-polynomial when $\lifetime$ is large, which occurs precisely when consecutive demands are timed far apart: for example, $D = \set{(u, v, 1), (v, u, 2^n)}$ yields $\lifetime = 2^n$ despite having polynomial encoding size.
We call such large time intervals \emph{gaps} and show that, once a gap is sufficiently long, any placement of the $k$ walks at the beginning of the gap can always be extended to any target placement at its end.
This allows us to replace long gaps in the flow network by a single unconstrained layer, bounding the total network size by $\poly(n, \cardinality D)$.
\begin{definition}[Relevant time step]
    \label{def:relevant-time}
    We call a time step $t \in \N$ \emph{relevant} if there are $u, v \in V(G)$ such that $(u, v, t) \in D$; otherwise we say $t$ is \emph{irrelevant}.
\end{definition}
\begin{definition}[$\gapsize$-gap]
    We call a time interval $I = [t, t+\gapsize)$ of size $\gapsize$ an \emph{$\gapsize$-gap} if all of the time steps in $I$ are irrelevant.
\end{definition}

Routing $k$ walks across a gap is an instance of the $\TEDW$ problem, where the start and end positions of the $k$ walks at the two boundaries of the gap form the terminal pairs:

\begin{problem}{Temporally Edge Disjoint Walks ($\TEDW$)}
    \Input & Temporal graph $\G = \langle G, \lambda \rangle$, terminal pairs $\set{(s_i, z_i) \in V(G)^2}_{i = 1}^k$. \\
    \Prob  & Are there temporally edge disjoint walks connecting each terminal pair $(s_i, z_i)$?
\end{problem}

In general, the $\TEDW$ problem is intractable.
\Textcite{klobas2023interference} show that the related $\TVDW$ problem is $\W[1]$-hard when parameterized by~$k$.
The standard vertex-to-edge reduction (splitting each vertex into an in-copy and an out-copy connected by a single edge) transfers this hardness to the edge disjoint variant, which yields:
\begin{proposition}\label{thm:tedw-hard}
    $\TEDW$ parameterized by $k$ is $\W[1]$-hard.
\end{proposition}
The $\TEDW$ instances arising from gaps, however, are special: since all time steps in a gap are irrelevant, and $\G(G,D)$ makes every edge of $G$ available at every time step, the underlying graph in which we are looking for the walks is effectively static throughout the gap.
The lemma below make this precise: any connected static subgraph $F \subseteq E(G)$ spanning all terminal pairs suffices to route $k$ walks TED through a gap, provided the gap is long enough.
\begin{lemma}[$((n-1) \cdot k)$-gap suffices]
    \label{thm:tedw-n-times-k}
    Let $\G = \langle G, \lambda \rangle$ be a temporal graph and $\set{(s_i, z_i)}_{i = 1}^k$ be a set of terminal pairs.
    Suppose that there is a set of edges $F \subseteq E(G)$ that connects each terminal pair and is present at all time steps until $\gapsize = (n-1) \cdot k$, that is, $t \in \lambda(e)$ for all $t \in [\gapsize], e \in F$.
    Then there are TED temporal paths $\set{p_i}_{i = 1}^k$ such that each $p_i$ is a shortest $s_i$-$z_i$-path in $F$.
\end{lemma}
\begin{proof}
    We construct the walks to move one after another, thus they will be temporally edge disjoint.
    For each $i$, let $p_i$ be a shortest $s_i$-$z_i$-path in $F$.
    We schedule it to start at time step $1 + (i-1) \cdot (n-1)$.
    As each path has length at most $n-1$, this does not overlap with any previous walk.
\end{proof}

We now define the resulting $\gapsize$-compressed static expansion and then equip it with the appropriate flow ranges to solve $\TEDSC$ in polynomial time.
This concept of $\gapsize$-compressed constructions will also play a vital role in the remainder of the paper.

\begin{definition}[$\gapsize$-compressed static expansion]
    \label{def:compressed-static-expansion}
    \label{def:compressed-tedsc-network}
    Let $(G,D)$ be a schedule completion instance.
    For every gap $[t_1 + 1, t_2)$ of size at least $\gapsize$ between two consecutive relevant time steps $t_1$ and $t_2$, we remove all layers from the $(t_1 + 2)$-th layer up to and including the $(t_2 - 1)$-th layer from $\SE(G,D)$.
    We add a new edge from $u_{t_1 + 1}$ to $v_{t_2}$ if $u$ can reach $v$ in $G$.
    The resulting graph is the \emph{$\gapsize$-compressed static expansion} $\SE_{\gapsize}(G,D)$.
    An example replacement is illustrated in \Cref{fig:gap-compression}.
    \begin{figure*}[t]
        \centering
        \includegraphics[height=7cm]{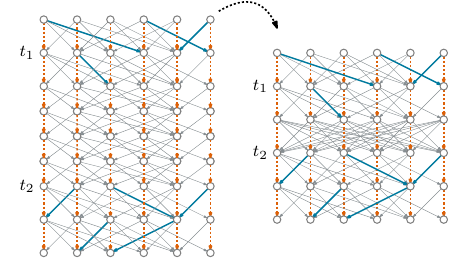}
        \caption{Illustration of the gap compression step. The left part shows some layers of a flow network with a large gap between the relevant time steps $t_1$ and $t_2$. The right part shows the replacement of these layers by a single compressed layer connecting all vertices from the $(t_1 + 1)$-th layer to all reachable vertices from the $t_2$-th layer.}
        \label{fig:gap-compression}
    \end{figure*}
    The final graph $\SE_{\gapsize}(G,D)$ is obtained by performing this step for all consecutive relevant time steps of sufficient size.

    The \emph{$\gapsize$-compressed schedule completion flow network} $(\SEcompressed{\gapsize}(G,D), r)$ is defined analogously to $(\SEflow(G, D), r)$, but using $\SE_\gapsize(G, D)$ as a basis.
    The newly inserted edges have unlimited flow range $(0, \infty)$.
\end{definition}

By computing the flow on the network of \Cref{def:compressed-tedsc-network} instead, we obtain a polynomial-time algorithm for $\TEDSC$.
\begin{theorem}
    \label{thm:tedsc-poly-n-d}
    $\TEDSC$ can be decided using one max-flow computation on a network with $\O(n \cdot \cardinality D \cdot (n \cdot k))$ vertices and $\O((n+m) \cdot \cardinality D \cdot (n \cdot k))$ edges.
\end{theorem}
\begin{proof}
    Set $\gapsize = (n-1) \cdot k$, construct $(\SEcompressed{\gapsize}(G,D), r)$ as in \Cref{def:compressed-tedsc-network}, and check for a feasible $s$-$z$-flow of value at most $k$.
    By \Cref{thm:tec-poly-n-tau}, a feasible flow of value at most $k$ in $(\SEflow(G,D), r)$ exists if and only if $(G, D, k)$ is a YES-instance.
    By \Cref{thm:tedw-n-times-k}, every gap of size at least $\gapsize$ can be compressed without changing feasibility, and hence the same holds for $(\SEcompressed{\gapsize}(G,D), r)$.

    It remains to bound the size of $\SEcompressed{\gapsize}(G,D)$.
    There are at most $\cardinality D$ relevant time steps, giving at most $\cardinality D - 1$ gaps between consecutive ones.
    Each uncompressed gap contributes fewer than $\gapsize$ layers; each compressed gap contributes no extra layers (its interior is replaced by edges between the two boundary layers).
    Together with the $\cardinality D$ relevant layers, $\SEcompressed{\gapsize}(G,D)$ has $\O(\cardinality D \cdot \gapsize)$ layers of $n$ vertices each, giving $\O(n \cdot \cardinality D \cdot \gapsize)$ vertices.
    Each uncompressed layer transition contributes $n$ waiting and $m$ move edges; each compressed gap contributes at most $n^2$ edges.
    Since $m \in \Omega(n)$, the gap edges are dominated, and the total edge count is $\O((n+m) \cdot \cardinality D \cdot (n \cdot k))$.
\end{proof}

\section{Hardness of Constrained TEDSC}
\label{sec:hardness}
In this chapter, we begin to investigate the complexity of constrained $\TEDSC$.
In contrast to unconstrained $\TEDSC$, both $\lenTEDSC$ and $\lifeTEDSC$ capture the complexity of notoriously hard problems.
This is in part because finding edge disjoint paths of bounded length in static graphs is already $\NP$-hard and we thus only need to remove the temporal aspect of $\TEDSC$ to capture this hardness.
Specifically, in \cref{sec:sat-reduction}, we show that a known hardness reduction for bounded length $\EDP$ also works for both $\lenTEDSC$ and $\lifeTEDSC$.
Next, in \cref{sec:edp-reduction} we can use the length and lifespan constraints to enforce a pairing of demands and thereby capture the hardness of $\EDP$.
This contrasts the results for unconstrained $\TEDSC$ where, although we are seeking edge disjoint paths in the static expansion $\SE(G, D)$, the demands can be fulfilled by any walk.
Finally, in \cref{sec:binpacking-reduction} we switch perspective and allow for a large number of demands that need to be fulfilled on a small network $G$.
We show that distributing the demands among the walks captures the $\BinPacking$ problem, in the sense that each bin corresponds to a walk and every item corresponds to a cluster of demands in the draft schedule.
Notably, this reduction works only for $\lenTEDSC$.
As we will see in \cref{sec:star-algorithm}, $\lifeTEDSC$ does not permit an equivalent reduction unless $\FPT = \W[1]$.

This chapter forms the basis for analyzing the parameterized complexity of constrained $\TEDSC$.
In combination with \cref{sec:exact}, we exhaustively explore the complexity landscape of the parameters $k$, $h$, and $\cardinality D$.
The reduction from $\BinPacking$ is the first step towards understanding the complexity of constrained $\TEDSC$ on restricted graph classes.
We return to this topic later in \cref{sec:star-algorithm}.

\subsection{Reduction from 3-SAT}
\label{sec:sat-reduction}
In this section, we present a reduction from $\TSAT$ to both $\lenTEDSC$ and $\lifeTEDSC$.
Our construction is based on the reduction by \textcite{itai1982complexity} for length bounded $\EDP$ in static graphs.
Notably, regardless of the input instance, we output an instance with a constant bound $h = 5$.
This implies that there cannot be an $\FPT$ or $\XP$ algorithm for $\lenTEDSC$ or $\lifeTEDSC$ parameterized by $h$ unless $\P = \NP$.
Furthermore, if ETH holds, the reduction yields a lower bound for the running time of any algorithm.

The first step of the reduction is the following normalization procedure for $\TSAT$ instances by \textcite{itai1982complexity}.
\begin{lemma}[\thmcite{itai1982complexity}]
    \label{thm:sat-var-occ-equal}
    Let $\varphi$ be an instance of $\TSAT$ with $n$ variables and $m$ clauses.
    There is a $\TSAT$ instance $\psi$ such that
    \begin{enumerate}[a)]
        \item $\psi$ is satisfiable if and only if $\varphi$ is,
        \item for each $i \in [n]$, the number of occurrences of $x_i$ and $\lnot x_i$ in $\psi$ are equal, and
        \item the size of $\psi$ is linear in the size of $\varphi$.\qedhere
    \end{enumerate}
\end{lemma}

\begin{theorem}
    \label{thm:h5-hard}
    Both $\lenTEDSC$ and $\lifeTEDSC$ are $\NP$-hard even if $h = 5$.
\end{theorem}
\begin{proof}
    We first show the hardness of $\lifeTEDSC$ and briefly argue how the proof translates to $\lenTEDSC$ at the end.
    Let $\varphi$ be a $\TSAT$ instance with $n$ variables and $m$ clauses.
    By \Cref{thm:sat-var-occ-equal}, we assume that for each variable $x_i$, the number of occurrences $n_i$ of $x_i$ and $\lnot x_i$ in $\varphi$ are equal.
    We construct a graph $G$ and a draft schedule $D$ such that there is a $D$-extending schedule $S$ of $4m$ walks with lifespan bound $h = 5$ if and only if $\varphi$ is satisfiable.

    $G$ is a union of $n$ graphs $G_1, G_2, \dots, G_n$ (one for each variable) and $m$ graphs $H_1, H_2, \dots, H_m$ (one for each clause).
    The graphs are illustrated in \Cref{fig:sat-reduction} and we describe them in the next two paragraphs.

    \begin{figure*}[t]
        \centering
        \begin{subfigure}{0.6\textwidth}
            \centering
            \includegraphics[height = 5.5 cm, page = 1]{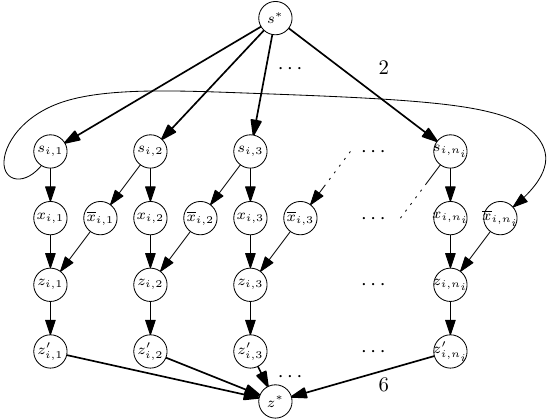}\\[1 ex]
            \caption{Graph $G_i$}
        \end{subfigure}%
        \begin{subfigure}{0.35\textwidth}
            \centering
            \includegraphics[height = 5.5 cm, page = 2]{assets/sat-reduction.pdf}\\[1 ex]
            \caption{Graph $H_j$}
        \end{subfigure}

        \caption{
            The components of $G$.
            Thick edges indicate demands at the labeled time step.
            The graph $G_i$ encodes the choice of a truth value for the variable $x_i$, which occurs $n_i$ times in $\varphi$.
            The graph $H_j$ encodes the constraint that at least one literal of the clause $c_j$ is satisfied.
            If $c_j$ contains the $k$-th occurrence of the variable $x_i$ in $\varphi$, the vertices are connected as shown.
            $H_j$ contains the same structure for the other two literals in $c_j$.
        }
        \label{fig:sat-reduction}
    \end{figure*}

    The vertices $s^*$ and $z^*$ are shared by every $G_i$ and $H_j$.
    For each $i \in [n]$, the graph $G_i$ contains five additional vertices for each occurrence of $x_i$ in $\varphi$.
    Formally, for every $k \in [n_i]$ there are vertices $s_{i,k}$, $x_{i,k}$, $\overline x_{i, k}$, $z_{i, k}$, and $z_{i, k}'$.
    The edges of $G_i$ consist of two paths of length $5$ for each $k \in [n_i]$: $s^* \, s_{i,k} \, x_{i,k} \, z_{i,k} \, z_{i,k}' \, z^*$ and $s^* \, s_{i,k+1} \, \overline x_{i,k} \, z_{i,k} \, z_{i,k}' \, z^*$ (where the addition in the index should wrap around so that $s_{i, n_i + 1} = s_{i, 1}$).
    Observe that one path uses $s_{i, k}$ and $x_{i, k}$ whereas the other uses $s_{i, k+1}$ and $\overline x_{i, k}$.

    For $j \in [m]$, the graph $H_j$ represents the $j$-th clause and contains two new vertices $s_j^c$ and $c_j$.
    It also contains edges that connect $c_j$ to the vertices corresponding to the variables in the $j$-th clause.
    Specifically, suppose that the $j$-th clause is $l_1 \lor l_2 \lor l_3$.
    For each literal $l \in \set{l_1, l_2, l_3}$, we add edges to $H_j$ as follows:
    If $l$ is the $k$-th occurrence of the non-negated variable $x_i$, then $H_j$ contains edges from $s_j^c$ to $s_{i, k}$ and from $x_{i, k}$ to $c_j$.
    Otherwise, if $l$ is the $k$-th occurrence of a negated variable $\overline x_i$, then we connect $s_j^c$ to $s_{i, k+1}$ and $\overline x_{i, k}$ to $c_j$ instead.
    Finally, $H_j$ contains edges connecting $s^*$ to $s_j^c$ and $c_j$ to $z^*$.

    To complete the $\lifeTEDSC$ instance, we set the lifespan bound $h = 5$, set the number of walks $k = 4m$, and construct the draft schedule as
    \begin{align*}
        D_{\mathrm{var}} & = \bigcup_{i \in [n], k \in [n_i]} \set{(s^*, s_{i, k}, 2), (z_{i, k}', z^*, 6)}, \\
        D_{\mathrm{cls}} & = \bigcup_{j \in [m]} \set{(s^*, s_j^c, 1), (c_j, z^*, 5)},                       \\
        D                & = D_{\mathrm{var}} \cup D_{\mathrm{cls}}.
    \end{align*}

    ($\Rightarrow$)\quad
    Suppose that there is a variable assignment $A\colon [n] \to \set{0,1}$ satisfying $\varphi$.
    Then we can construct a schedule as follows:
    For each variable $i \in [n]$, $n_i$ walks start at $s^*$ at time step $2$.
    We refer to the set of all these walks as \emph{variable walks}.
    Suppose that $A(i) = 1$, that is, $x_i$ is true in $A$.
    Then the $k$-th walk uses the path $s^* \, s_{i, k+1} \, \overline x_{i, k} \, z_{i, k} \, z_{i, k}' \, z^*$ such that it crosses the $z_{i, k}' z^*$ edge at time step $6$.
    If instead $A(i) = 0$, then the path $s^* \, s_{i, k} \, x_{i, k} \, z_{i, k} \, z_{i, k}' \, z^*$ should be used.
    These paths are internally vertex disjoint, so they are edge disjoint.

    For each clause $j \in [m]$, there is one additional walk that starts on $s^*$ at time step $1$.
    We call these walks \emph{clause walks}.
    As $A$ is a satisfying assignment, at least one of the literals in the $j$-th clause is satisfied, so let $l$ be any such literal.
    Suppose that $l$ is the $k$-th occurrence of the non-negated variable $x_i$.
    Then the walk takes the path $s^* \, s_j^c \, s_{i, k} \, x_{i, k} \, c_j \, z^*$.
    If $l$ is the $k$-th occurrence of a negated variable $\overline x_i$, then the walk takes the path $s^* \, s_j^c \, s_{i, k+1} \, \overline x_{i, k} \, c_j \, z^*$.
    This walk is edge disjoint with all variable walks because the edge $s_{i, k} x_{i, k}$ ($s_{i, k+1} \overline x_{i, k}$) is used only if $x_i$ is false (true) in $A$.
    It is disjoint to other clause walks because all clause walks are internally vertex disjoint.
    Finally, all described walks have length and lifespan $5$.
    As the total number of literals is $3m$, the total number of walks is $4m$.

    ($\Leftarrow$)\quad
    Suppose, on the other hand, that there is a feasible schedule $S$ of size at most $4m$.
    Observe that by the construction of $D$, these walks must follow the patterns of the variable and clause walks described previously:
    First, each walk can cover at most two demands and it can do so in exactly one of two ways.
    Either it covers one demand at time step $1$ and one at time step $5$, or alternatively it covers one at time step $2$ and one at time step $6$.
    The other combinations are not possible because the second required edge is too far away to reach in time.
    As there are at most $4m$ walks and $8m$ demands, all $4m$ walks must be used, and each must cover exactly two demands.
    Observe that the distances between the demand pairs require that every walk moves at each time step.
    The walks covering $D_{\mathrm{var}}$ must follow the pattern of the variable walks.
    To reach a second demand in time, these walks must take either the $x_{i, k} z_{i, k}$ or the $\overline x_{i, k+1} z_{i, k}$ edge.
    This means that the walk covering the demand $s^* s_{i, k+1}$ cannot take the respective other edge, because they would meet at $z_{i, k}$ at the same time and one of them could not reach a demand in time while also staying TED.
    Based on this, we define the variable assignment $A$ where each variable $x_i$ is true in $A$ if and only if the edge $x_{i, 1} z_{i, 1}$ is \emph{not} taken by any walk.

    To prove that $A$ satisfies $\varphi$, observe that there are $m$ more walks in $S$ which cover $D_{\mathrm{cls}}$.
    The walk that covers the demand $(c_j, z^*, 5)$ must also cover a second demand from $D_{\mathrm{cls}}$.
    Thus, it must be of the form $s^* \, s_{j'}^c \, s_{i, k} \, x_{i, k} \, c_j \, z^*$ or $s^* \, s_{j'}^c \, s_{i, k+1} \, \overline x_{i, k} \, c_j \, z^*$ for some $i \in [n], j' \in [m]$.
    Notably, the edges $s_{i, k} x_{i, k}$ or $s_{i, k+1} \overline x_{i, k}$ are taken at time step $3$, as all walks move at each time step.
    This also means that the former (latter) can be taken only if $x_i$ is true (false) in $A$, as it needs to be TED from the variable walks.
    As the edge $x_{i, k} c_j$ ($\overline x_{i, k} c_j$) exists only if $x_i$ ($\overline x_i$) is a literal in the $j$-th clause, this clause is satisfied by $A$.

    This concludes the proof for $\lifeTEDSC$.
    To prove that this reduction also works for $\lenTEDSC$, observe that any walk covering two demands needs to move at every time step, as it can otherwise not reach a second demand.
    Thus, any feasible solution to the $\lenTEDSC$ instance is also a $\lifeTEDSC$ solution and thus implies a $\TSAT$ solution.
\end{proof}

By tracking the size of the parameters, the above reduction yields a conditional lower bound based on ETH.

\begin{corollary}
    Unless ETH fails, there is no algorithm that solves $\lenTEDSC$ or $\lifeTEDSC$ in $2^{o(n+m+k+\cardinality D)}$ time.
\end{corollary}
\begin{proof}
    First, modifying the instance according to \Cref{thm:sat-var-occ-equal} increases the size of the formula only linearly.
    Then, for every variable and clause of the given formula $\varphi$, a constant number of vertices, edges, walks, and demands are added to the $\lenTEDSC$ or $\lifeTEDSC$ instance.
    Therefore, if an algorithm exists that solves $\lenTEDSC$ or $\lifeTEDSC$ in the desired time, we can solve $\TSAT$ in sub-exponential time, which contradicts \cref{thm:eth-nm}.
\end{proof}

\subsection{Reduction from EDP on DAGs}
\label{sec:edp-reduction}
In this section, we reduce $\EDP$ on static graphs to constrained $\TEDSC$.
As discussed in \Cref{sec:related-work}, the literature on the parameterized complexity of $\EDP$ is very rich.
We noted there that we can reduce the vertex and edge variants of the disjoint path problem to each other in polynomial time.
Although this fact is generally known, we provide a proof here for completeness.

\begin{lemma}
    There are polynomial time algorithms $f_{V \to E}$ and $f_{E \to V}$ that transform between $\VDP$ and $\EDP$.
    In particular, let $\mathcal I = (G, \set{(s_i, t_i) \in V(G)^2}_{i = 1}^k)$ be an instance of $\VDP$ ($\EDP$).
    Then $f_{V \to E}(\mathcal I)$ ($f_{E \to V}(\mathcal I)$) is an equivalent instance of $\EDP$ ($\VDP$) with the same number of terminal pairs.
    Moreover, the number of vertices in the graphs of $f_{V \to E}(\mathcal I)$ and $f_{E \to V}(\mathcal I)$ is $2 \cdot \cardinality{V(G)}$ and $k \cdot \cardinality{V(G)} + \cardinality{E(G)}$, respectively.
\end{lemma}
\begin{proof}
    For $f_{V \to E}$, suppose you want to find $k$ vertex disjoint paths in a graph $G$, then construct a graph $H$ as follows:
    For every vertex $v \in V(G)$, introduce two vertices $v^{\mathrm{in}}$ and $v^{\mathrm{out}}$ connected by a directed edge $(v^{\mathrm{in}}, v^{\mathrm{out}})$ in $H$.
    Then add the edges $(u, v) \in E(G)$ as edges $(u^{\mathrm{out}}, v^{\mathrm{in}})$ to $H$.
    Any set of vertex disjoint $(s_i, t_i)$ paths in $G$ correspond to edge disjoint $(s_i^{\mathrm{in}}, t_i^{\mathrm{out}})$ paths in $H$ and vice versa.

    For $f_{E \to V}$, suppose you want to find $k$ edge disjoint paths in a graph $G$, then modify your input to a graph $H$ as follows:
    For every vertex $v \in V(G)$, $H$ has $k$ vertices $v^1, v^2, \dots, v^k$.
    For each edge $(u, v) \in E(G)$, introduce one vertex $e^{uv}$ and $2k$ edges $(u^i, e^{uv}), (e^{uv}, v^i)$.
    The $i$-th of the $k$ vertex disjoint paths uses the vertices $\set{v^i}_{v \in V(G)}$ corresponding to vertices in $G$, and the vertices $\set{e^{uv}}_{(u, v) \in E(G)}$ corresponding to edges in $G$.
    Thus, edge disjoint paths in $G$ that connect $\set{(s_i, t_i)}_{i = 1}^k$ are exactly vertex disjoint paths in $H$ that connect $\set{(s_i^i, t_i^i)}_{i = 1}^k$.
    Note that the number of paths $k$ one can ask for is at most $\cardinality{E(G)}$, so this transformation is polynomial.
\end{proof}
\Textcite{klobas2023interference} show that finding temporally vertex disjoint walks ($\TVDW$) is $\W[1]$-hard even on DAGs by reducing from $\EDP$ on static DAGs~\cite{slivkins2010parameterized}.
This does not contradict our result that $\TEDSC \in \P$ because in $\TEDSC$ the demands are ubiquitous, whereas in $\TVDW$ the demands come paired in terminals.
For the same reason, single-commodity flow is tractable, whereas two-commodity flow is not.

However, once we consider $\lenTEDSC$ or $\lifeTEDSC$, we have another dimension available to restrict the shape of the temporal walks, which allows us to show the same hardness result.
In the following reduction, we first subdivide a static DAG so that static and temporal edge disjointness become equivalent.
This step is similar to the reduction by \textcite{klobas2023interference}.
In contrast to their reduction, though, we cannot output an arbitrary temporal graph, but have to adhere to the fixed temporal structure of $\G(G, D)$ in $\TEDSC$.
Going beyond their reduction, we then use the length/lifespan constraints to enforce a pairing of demands corresponding to the pairing of terminal pairs of a given $\EDP$ instance.
As this second step assumes the graph structure constructed in the first step, we also reduce from the static $\EDP$ problem rather than using $\TVDW$ as an intermediate.
Finally, both the reduction by \textcite{klobas2023interference} and our result can be strengthened by reducing from $\EDP$ on \emph{planar} DAGs.
The corresponding hardness has been shown by \textcite{chitnis2023tight}, and both the reduction by \textcite{klobas2023interference} and our reduction preserve planarity.

\begin{theorem}
    \label{thm:hard-by-D-planar-dag}
    Both $\lenTEDSC$ and $\lifeTEDSC$ parameterized by $k + \cardinality D$ are $\W[1]$-hard on planar DAGs.
\end{theorem}
\begin{proof}
    We reduce from $\EDP$ on planar DAGs parameterized by the number of terminal pairs~\cite{chitnis2023tight} to $\lenTEDSC$.
    Using \Cref{thm:len-life-equiv}, it is easy to see that the proof works analogously for $\lifeTEDSC$.
    Given the graph $G$ and the terminal pairs $(s_1, z_1), \dots, (s_k, z_k)$, let $v^1, v^2, \dots, v^n$ be a topological ordering, that is, an ordering such that for every $(v^i, v^j) \in E(G)$ we have $i < j$.
    For all $v^i$, define $\varphi(v^i) = i$.
    We construct a graph $G'$ that is initially equal to $G$ and then modify it as described in the following two paragraphs.
    An example is illustrated in \Cref{fig:edp-reduction}.

    \begin{figure*}[t]
        \centering
        \includegraphics[height=10cm]{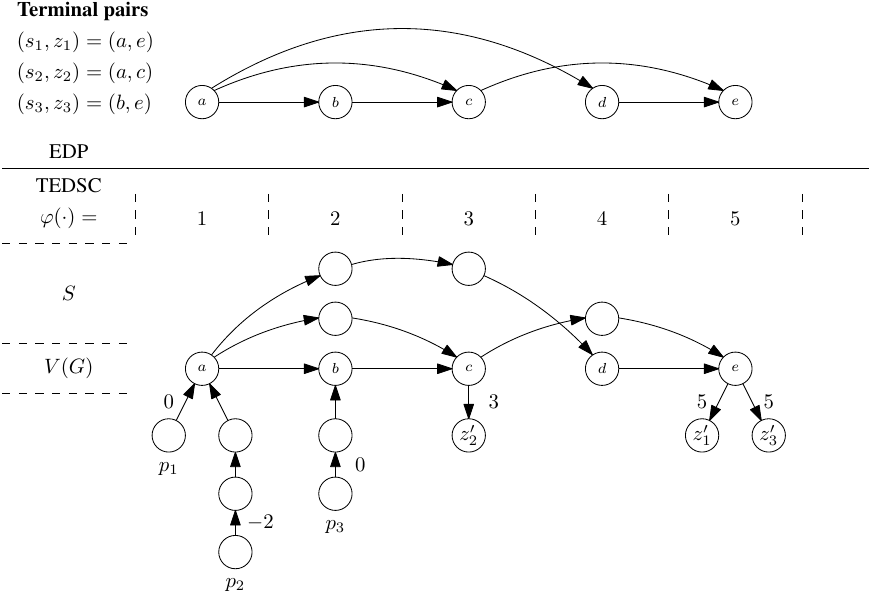}
        \caption{
            Example instance with $n = 5$ vertices, which gives $h = 6$.
            Numbers next to edges indicate a demand at that time.
            A possible solution for the $\EDP$ instance consists of the paths $\set{ade, ac, bce}$.
        }
        \label{fig:edp-reduction}
    \end{figure*}

    First, for every edge $(v^i, v^j) \in E(G)$ with $j - i > 1$, we subdivide it in $G'$ such that the resulting path contains $j - i$ edges.
    Thus, after transforming all edges, every $v^i$-$v^j$-path in $G'$ has length $j - i$.
    Let $S$ denote the set of vertices created for subdividing edges.
    We extend the definition of $\varphi$ to $S$ such that if the edge $(u, v)$ is subdivided as $u = u_0 u_1 u_2 \dots u_\ell = v$, then $\varphi(u_i) = \varphi(u) + i$.
    Set $h = n + 1$ and $k$ as given in the $\EDP$ instance.

    Next, let $(s_i, z_i)$ be a terminal pair.
    We add a vertex $z_i'$ with a single incident edge $(z_i, z_i')$ to $G'$.
    We will refer to those edges as \emph{exit edges}.
    Furthermore, we attach a path $p_i$ of length $\ell(p_i) = h - (\varphi(z_i) - \varphi(s_i)) - 1$ to $G'$ such that the last vertex of $p_i$ is $s_i$ and all other vertices of $p_i$ are created new.
    Let $(a, b)$ be the first directed edge of $p_i$.
    The distance between $a$ and $z_i'$ is thus exactly $h$.
    Finally, we add the temporal edges $(a, b, \varphi(z_i) - h + 1)$ and $(z_i, z_i', \varphi(z_i))$ to $D$.
    Observe that all modifications applied to $G$ to construct $G'$ preserve planarity, so $G'$ is a planar DAG.

    ($\Rightarrow$)\quad
    Suppose that there is a solution to the $\EDP$ instance.
    Then the paths should start in the corresponding $p_i$ paths, travel into $V(G)$, and then exit via the edge $(z_i, z_i')$.
    The trip should start at the time of the demand edge on $p_i$, which leaves exactly $h$ time steps until the path is on $z_i$ to take the demand edge to $z_i'$ at the demanded time.
    Since the paths are edge disjoint in $G$ and all edges outside of $E(G)$ are traversed by only a single path, the resulting paths are temporally edge disjoint.

    ($\Leftarrow$)\quad
    Suppose, on the other hand, that the constructed $\TEDSC$ instance is feasible.
    We show that translating the temporal walks from $V(G')$ to $V(G)$ and viewing them as static rather than temporal walks gives a solution to the $\EDP$ instance.
    \begin{claim}
        Every walk fulfills exactly two demands, one on a path $p_i$ and one on an exit edge $(z_i, z_i')$.
    \end{claim}
    Every exit edge of $G'$ leads to a sink, so after taking an exit edge, a walk cannot fulfill any further demands.
    Analogously, every demand on a path $p_i$ is on an edge from a source of $G'$.
    As there are $2k$ demands in total, every walk must fulfill exactly two demands, one on a path $p_i$ and one on an exit edge.
    \begin{claim}
        All walks have length $h$.
    \end{claim}
    It suffices to show that the total distance traveled by all walks is $kh$.
    Each walk traverses one of the paths $p_i$ and one of the exit edges.
    Thus, we need to show that the walks travel $kh - \sum_{i = 1}^k \len(p_i) - k = \sum_{i = 1}^k \varphi(z_i) - \sum_{i = 1}^k \varphi(s_i)$ steps inside of $V(G) \cup S$.
    We do so using a potential argument:
    Suppose that the walks are on the vertices $u_1, u_2, \dots, u_k \in V(G) \cup S$, then the potential is defined as $\varPhi = \sum_{i = 1}^k \varphi(u_i)$.
    When all walks have traversed only the paths $p_i$ and are on the $s_i$ vertices then $\varPhi = \sum_{i = 1}^k \varphi(s_i)$.
    Before the walks leave $V(G) \cup S$ via the exit edges we have $\varPhi = \sum_{i = 1}^k \varphi(z_i)$.
    The potential increases by $1$ exactly when a walk traverses an edge in $V(G) \cup S$.
    Therefore, exactly the desired number of steps is taken.
    \begin{claim}
        The walks connect the terminal pairs $(s_i, z_i)$.
    \end{claim}
    After entering $V(G)$ through a path $p_i$, a walk must travel inside $V(G) \cup S$ for the rest of its lifespan before taking an exit edge.
    By construction, the remaining lifespan after entering through $p_i$ is $h - (\varphi(z_i) - \varphi(s_i))$, so it must take $\varphi(z_i) - \varphi(s_i)$ edges inside $V(G) \cup S$.
    As $p_i$ ends on $s_i$, this leads exactly to the exit edge $(z_i, z_i')$.
    \begin{claim}
        The walks are edge disjoint.
    \end{claim}
    We know that the paths are temporally edge disjoint by assumption.
    The previous claims and the chosen time steps of the demands imply that all walks are constantly moving during their lifespan.
    By construction of $D$, the only time a walk can be on a vertex $v \in V(G) \cup S$ is at the time step $\varphi(v)$.
    Thus, for every edge inside $V(G) \cup S$, no two walks cross this edge at two different times.
    This implies that the paths are not only temporally edge disjoint, but also edge disjoint and therefore correspond to a solution of the $\EDP$ instance.
\end{proof}

\Textcite{chitnis2023tight} also proves a lower bound of $f(k) n^{o(k)}$ for $\EDP$ on planar DAGs based on ETH.
By tracking the modifications from our reduction, we get the following result:

\begin{corollary}
    Unless ETH fails, there is no algorithm that solves $\lenTEDSC$ or $\lifeTEDSC$ in $f(k + \cardinality D) n^{o(k + \cardinality D)}$ time.
\end{corollary}
\begin{proof}
    We denote by $k' = k + \cardinality D = 3k$ the parameter of the constructed $\lenTEDSC$ or $\lifeTEDSC$ instance.
    Further, let $n' = \cardinality{V(G')}$.
    By subdividing the edges of $G$, we add at most $m \cdot n \in \O(n^3)$ vertices to $G'$.
    The attached paths contain at most $k \cdot h \le m \cdot n \in \O(n^3)$ vertices.
    Thus, if an algorithm could solve $\lenTEDSC$ or $\lifeTEDSC$ in time $f(k') (n')^{o(k')}$, then it would solve the underlying $\EDP$ instance in time $f(3k) \O(n^3)^{o(3k)} = f'(k) n^{o(k)}$, which contradicts the lower bound by \textcite{chitnis2023tight}.
\end{proof}

\subsection{Reduction from Bin Packing}
\label{sec:binpacking-reduction}
In this section, we show $\W[1]$-hardness of $\lenTEDSC$ by reducing from the $\BinPacking$ problem.
The $\BinPacking$ problem is a generalization of the $\Partition$ problem.

\begin{problem}[]{Bin Packing}
    \Input & Universe $U$ of $n$ items, size function $s\colon U \to \N^+$, integers $k, B \in \N^+$. \\
    \Prob  & Does there exist a partition of $U$ into $k$ sets $U_1, U_2, \dots, U_k$ such that $s(U_j) = \sum_{x \in U_j} s(x) \le B$ for all $j \in [k]$?
\end{problem}

\Textcite{jansen2013bin} show that even when all integers in the input (including the item sizes) are encoded in unary, the problem is still $\W[1]$-hard when parameterized by the number of bins $k$.
We refer to this variant of the problem as $\UnaryBinPacking$.
Note that $\BinPacking$ parameterized by the bin size $B$ is in $\FPT$~\cite{doring2025parameterized}.

The transformation from $\BinPacking$ to $\lenTEDSC$ is very natural:
A bin of capacity $B$ is the same as a walk that can make $B$ steps.
In that sense, if an item $x$ is packed into a bin, then the corresponding walk must travel a distance of $s(x)$.
An easy way to realize this setup is to build a tree whose root initially hosts all walks and whose branches are paths of length $s(x)$.
By placing demands at the ends of the paths at sufficiently spaced time steps, traversing the branches is equivalent to placing the corresponding item into the corresponding bin.
This scheme is very similar to that of \textcite{doring2025parameterized} for capacitated vehicle routing.
As the item sizes are encoded in unary, constructing the branches with $s(x)$ vertices is still polynomial in the size of the instance.

Observe that each branch of the tree is used exactly once by a single walk.
If we place demands at the root of the tree that force all walks to return in between handling the demands of the items, we can merge the branches of the tree into a single path.
Thus, the walks are hosted at one end of the path, and each item appears as a demand somewhere on the path.
\begin{proposition}
    $\lenTEDSC$ parameterized by $\cardinality D$ is $\W[1]$-hard, even if $G$ is a bidirected path.
\end{proposition}

Finally, we can simplify this construction even further by collapsing the entire path into just two vertices.
To ensure that each item $x$ still requires $s(x)$ steps, we add multiple demands to each edge.
In particular, there are $2 s(x)$ demands that direct a walk back and forth.
To ensure that the same walk fulfills all these demands, we introduce an additional vertex that serves as the waiting area for all other walks.

\begin{theorem}
    \label{thm:length-w-hard-on-p3}
    $\lenTEDSC$ parameterized by $k$ is $\W[1]$-hard, even if $G$ is a bidirected path of $3$ vertices.
\end{theorem}
\begin{proof}
    We reduce from $\UnaryBinPacking$, so let $(U, s, k, B)$ be an instance.
    Assume without loss of generality that $\sum_{x \in U} s(x) = k \cdot B$.
    Let $G$ be the bidirected path graph with $V(G) = \set{u, v, w}$ and $E(G) = \set{(u, v), (v, u), (v, w), (w, v)}$.
    We output a $\lenTEDSC$ instance $\mathcal I = (G, D, k, h = 2B + 1)$ where the draft schedule $D$ is composed as described in the following paragraph.
    The construction is illustrated in \cref{fig:binpacking-reduction}.

    \begin{figure*}[t]
        \centering
        \begin{subfigure}{0.45\textwidth}
            \centering
            \includegraphics[height=7.5cm, page = 1]{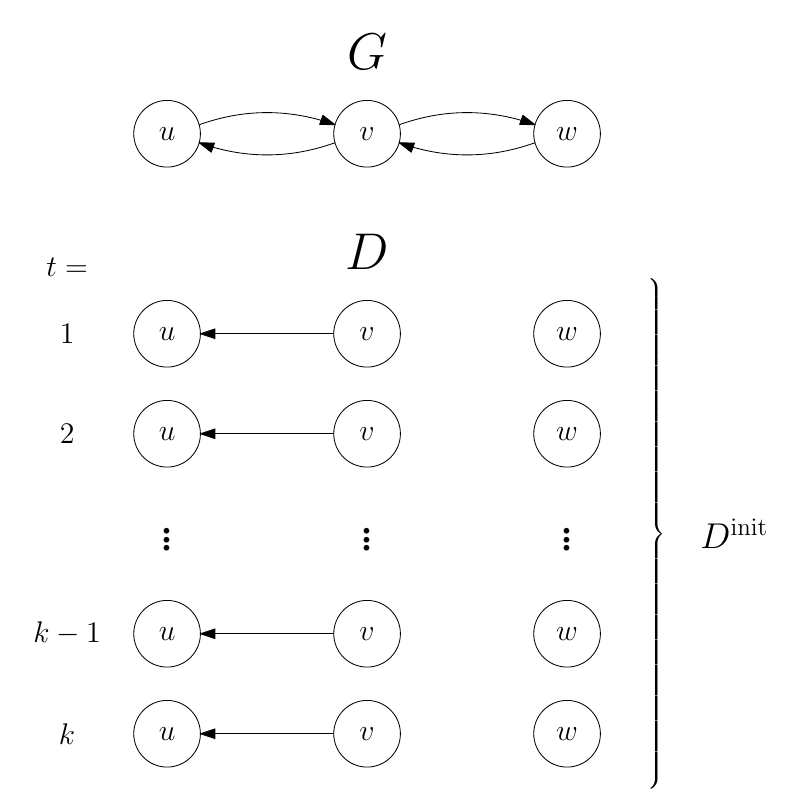}\\[1 ex]
            \caption{Graph $G$ and initialization phase.}
        \end{subfigure}%
        \begin{subfigure}{0.5\textwidth}
            \centering
            \includegraphics[height=7.5cm, page = 2]{assets/binpacking-reduction.pdf}\\[1 ex]
            \caption{Phase of item $x_i$.}
        \end{subfigure}%
        \caption{Demands of each phase in the constructed instance.}
        \label{fig:binpacking-reduction}
    \end{figure*}
    The demands span $k + \sum_{x \in U} 2 s (x)$ time steps, and at every time step in this range, there is exactly one demand.
    The demands are grouped into $1 + \cardinality U$ phases.
    First, there is an initialization phase that consists of $k$ successive demands on the edge $(v, u)$.
    Formally, the corresponding demands are $D^{\mathrm{init}} = \set{(v, u, t) \with t \in [k]}$.
    Next, there are $\cardinality U$ phases, one for each item of $U$.
    The phase of an item $x \in U$ takes $2 s(x)$ time steps.
    Let $\set{x_1, x_2, \dots, x_{\cardinality U}} = U$ be an ordering of the items and for each $i \in [\cardinality U]$ let $b(i) = k + 1 + \sum_{j = 1}^{i-1} 2 s(x_j)$ be the beginning time of the $i$-th phase.
    At the beginning of the phase of item $x_i$, there is a single demand $d_i^{\mathrm{start}} = (u, v, b(i))$ from $u$ to $v$.
    Similarly, at the end of the phase, there is a demand $d_i^{\mathrm{end}} = (v, u, b(i + 1) - 1)$ to return from $v$ to $u$.
    During the phase, there are $2 \cdot (s(x_i) - 1)$ demands
    \begin{align*}
        D_i^{\mathrm{forth}} & = \set{(v, w, b(i) + 1 + 2j) \with j = 0, 1, \dots, s(x_i) - 2}, \\
        D_i^{\mathrm{back}}  & = \set{(w, v, b(i) + 2 + 2j) \with j = 0, 1, \dots, s(x_i) - 2}, \\
        D_i^{\mathrm{work}}  & = D_i^{\mathrm{back}} \cup D_i^{\mathrm{forth}}
    \end{align*}
    going back and forth between $v$ and $w$.
    We denote $D_i = \set{d_i^{\mathrm{start}}, d_i^{\mathrm{end}}} \cup D_i^{\mathrm{work}}$.
    The draft schedule comprises the demands of all phases
    \[
        D = D^{\mathrm{init}} \cup \bigcup_{x_i \in U} D_i.
    \]

    ($\Rightarrow$)\quad
    Observe that if there is a partition of $U$ into sets $\set{U_i}_{i = 1}^k$ of size at most $B$, then we can construct a feasible schedule for $\mathcal I$ as follows.
    For each $i$, there is one walk that first covers the $i$-th demand of $D^{\mathrm{init}}$ and then covers the demands $D_j$ for all $j \in U_i$.
    As $s(U_i) \le B$, the length of this walk is at most $h$.

    ($\Leftarrow$)\quad
    Suppose on the other hand there is some schedule $S = \set{p_i}_{i = 1}^k$ for $\mathcal I$.
    First, observe that as $s(U) = k \cdot B$, we have $\cardinality D = k + 2 \cdot s(U) = k \cdot h$, so a walk of $S$ can move only when the temporal edge it is taking is in $D$.
    This implies that every walk in $S$ covers exactly one demand in $D^{\mathrm{init}}$.
    Thus, at time step $k$, all $k$ walks are on vertex $u$.
    Next, observe that only one walk can move during each item phase:
    To cover a demand of $D_i^{\mathrm{work}}$, the walk must be on $v$ or $w$, which is only possible if this walk covered $d_i^{\mathrm{start}}$.
    Furthermore, this same walk must then also cover $d_i^{\mathrm{end}}$.
    At the end of each phase, all walks are on vertex $u$.
    As all walks have length $h = 2B + 1$ and $\cardinality{D_i} = 2 s(x_i)$ for all $i$, the partition $\set{U_i = \set{j \with d_j^{\mathrm{start}} \in \E(p_i)}}_{i = 1}^k$ satisfies $s(U_i) = B$ for all $i \in [k]$.
\end{proof}

A consequence of \Cref{thm:length-w-hard-on-p3} is that there is no $\XP$ algorithm for $\lenTEDSC$ parameterized by $n$ unless $\P = \NP$.
Furthermore, this holds even when the instances are restricted to simple graphs like paths or star graphs.

\begin{corollary}
    $\lenTEDSC$ parameterized by $n$ is $\paraNP$-hard, even if $G$ is a bidirected path or star.
\end{corollary}

\section{Exact Algorithms for Constrained TEDSC}
\label{sec:exact}
In this chapter, we present exact algorithms for constrained $\TEDSC$ that complement the hardness results of \cref{sec:hardness}.
First, we complete the landscape of parameterized complexity of constrained $\TEDSC$ with regard to both $\FPT$ and $\XP$ for all combinations of the parameters $k$, $h$, and $\cardinality D$.
Recall that we assume $k < \cardinality D \le k \cdot h$ by \Cref{thm:k-lt-D,thm:D-le-kh}.
Regarding membership in $\FPT$, we know from \cref{sec:hardness} that, unless $\FPT = \W[1]$, there is no $\FPT$ algorithm for constrained $\TEDSC$ parameterized by $h$ or $k + \cardinality D$.
Thus, the next smallest parameter to investigate for possible $\FPT$ algorithms is $k + h$.
We construct such an algorithm in \cref{sec:fpt-kh}.
Regarding membership in $\XP$, we showed $\paraNP$-hardness by $h$ in \cref{sec:sat-reduction}, which rules out $\XP$ algorithms unless $\P = \NP$.
For $k$ and $\cardinality D$, however, we showed only $\W[1]$-hardness, which does not provide any information about the existence of $\XP$ algorithms.
We resolve this in \cref{sec:k-xp} where we present an $\XP$ algorithm for constrained $\TEDSC$ parameterized by $k$.
As $k < \cardinality D$, this also works as an $\XP$ algorithm for the parameter $\cardinality D$.

Second, in \cref{sec:star-algorithm}, we return to studying limited graph classes.
We saw in \cref{thm:length-w-hard-on-p3} that $\lenTEDSC$ parameterized by $k$ is $\W[1]$-hard even when $G$ is a bidirected path of three vertices.
However, unlike our other reductions, this result does not translate to $\lifeTEDSC$, because it relies on the fact that in $\lenTEDSC$ walks can wait on a vertex indefinitely.
In \cref{sec:star-algorithm}, we show that these long waiting periods are indeed a source of complexity by proving the opposite statement of \cref{thm:length-w-hard-on-p3} for $\lifeTEDSC$.
Formally, we show that there is a polynomial time algorithm for $\lifeTEDSC$ on any \emph{fixed} bidirected star graph (which includes the bidirected $P_3$ from \cref{thm:length-w-hard-on-p3}).
We achieve this by limiting the search space explored by the $\XP$ algorithm from \cref{sec:k-xp}.

\subsection{FPT Algorithm using Bounded Length EDP}
\label{sec:fpt-kh}
By \Cref{thm:hard-by-D-planar-dag}, there is no $\FPT$ algorithm for $\lenTEDSC$ or $\lifeTEDSC$ parameterized by $k + \cardinality D$ unless $\FPT = \W[1]$.
As $\cardinality D \le k \cdot h$, the parameter $k + \cardinality D$ is bounded by a function of $k + h$.
Thus, after ruling out $\FPT$ algorithms for the parameters $k + \cardinality D$ and $h$, the natural next question is whether an $\FPT$ algorithm exists for the parameter $k + h$.
In this section, we answer this question positively by reducing to $\lenEDP$.

\begin{problem}[]{\lang{Bounded Length Edge Disjoint Paths} ($\lenEDP$)}
    \Input & Graph $G$, terminal pairs $\set{(s_i, t_i)}_{i = 1}^k$, length bound $h \in \N$. \\
    \Prob  & Are there edge disjoint paths $\set{p_i}_{i = 1}^k$ such that $p_i$ is an $s_i$-$t_i$-path and $\len(p_i) \le h$ for all $i \in [k]$?
\end{problem}

\Textcite{golovach2011paths} present an $\FPT$ algorithm for $\lenEDP$ parameterized by $k + h$.
This result is the basis of the algorithms presented in this section.

\begin{theorem}[\thmcite{golovach2011paths}]
    \label{thm:len-edp-fpt-kh}
    There is an algorithm that decides $\lenEDP$ in $2^{\O(k \cdot h)} \log n \cdot m \cdot k$ time.
\end{theorem}

On a high level, finding a set of paths satisfying a $\lenEDP$ instance is similar to finding a set of paths in the static expansion of the graph of a constrained $\TEDSC$ instance.
In practice, however, there are two major differences we need to address.
First, in $\lenEDP$ the paths need to always be edge disjoint, but in $\TEDSC$ we allow multiple paths to traverse the same \emph{waiting edge} in $\SE(G, D)$ (recall that we gave waiting edges infinite capacity in \cref{def:tedsc-flow-network}).
Second, in $\lenEDP$ each path has to connect a specific pair of terminals, whereas in $\TEDSC$ the paths need to jointly cover the edges that correspond to the demands in $D$.
We resolve these differences as follows:
\begin{enumerate}
    \item We define the \emph{parallelize-subdivide} ($\PS$) operation, which simulates giving waiting and compressed edges infinite capacity, similar to \cref{def:tedsc-flow-network}.
    \item
          We generate $\lenEDP$ instances on $\SE(G, D)$ in which every path describes how to get from one demand edge to the next.
          This means that one temporal walk corresponds to multiple terminal pairs in the $\lenEDP$ instance.
\end{enumerate}
For $\lifeTEDSC$, this works out nicely because the length of a path in $\SE(G, D)$ equals the lifespan of the corresponding temporal walk.
To make this algorithm work for $\lenTEDSC$, we need a slightly more complicated approach.
\begin{enumerate}
    \addtocounter{enumi}{2}
    \item We generalize the algorithm of \cref{thm:len-edp-fpt-kh} to allow different length bounds among the terminal pairs.
    \item We define the In-Out expansion $\IO(\G)$ as an alternative to $\SE(\G)$ that preserves the length of temporal walks rather than their lifespan.
\end{enumerate}

We begin by defining the $\PS$ operation.
The purpose of this operation is to modify the graph of a $\lenEDP$ instance such that it simulates a flow network where some edges have infinite capacity.
If we are looking for $k$ paths, then infinite capacity is equivalent to a capacity of $k$, which we model by adding $k$ parallel edges.
\begin{definition}[$\PS$]
    \label{def:ps}
    Let $G$ be a static graph, $k \in \N$, and $E_\infty$ a set of edges to which the operation should be applied.
    The graph $\PS(G, k, E_\infty)$ is obtained from $G$ as follows:
    Replace every edge $e \in E_\infty$ with $k$ parallel copies of $e$.
    Then, subdivide every edge in the graph, that is, replace each edge $(u, v)$ with a new vertex $x$ and two edges $(u, x)$ and $(x, v)$.
\end{definition}
Note that if $k > 1$, then the first step of $\PS$ turns the graph into a multi-graph.
In the second step, every edge is replaced with a new distinct vertex $x$, so the final graph does not contain multi-edges anymore.

\begin{lemma}
    \label{thm:ps-paths-exist-iff}
    Let $G$, $k$, and $E_\infty$ be defined as in \cref{def:ps} and let $\set{(s_i, t_i) \in V(G)^2}_{i = 1}^k$ be a set of terminal pairs.
    The following statements are equivalent:
    \begin{itemize}
        \item There exists a set of paths $\set{p_i}_{i = 1}^k$ in $G$ that are edge disjoint on $E(G) \setminus E_\infty$ and connect the terminal pairs.
        \item There exists a set of paths $\set{p_i'}_{i = 1}^k$ in $G' = \PS(G, k, E_\infty)$ that are edge disjoint on all of $E(G')$ and connect the terminal pairs.
    \end{itemize}
    Moreover, if these sets exist, then there is a correspondence between the paths such that $\len(p_i') = 2 \len(p_i)$ for all $i \in [k]$.
\end{lemma}
\begin{proof}
    Given the paths $\set{p_i}_{i = 1}^k$, we construct each $p_i'$ such that for each edge $e$ of $p_i$, $p_i'$ follows the two edges created by subdividing $e$ in $\PS$.
    Specifically, the $i$-th path $p_i'$ should take the $i$-th of the $k$ copies of $e$.
    This establishes all desired properties.
    For the reverse, note that as the terminal pairs are vertices in $G$, every path $p_i$ alternates between edges to and from vertices created during the subdivision step of $\PS$.
    The paths $p_i$ can be constructed by following the same sequence of vertices as $p_i'$, but limited to $V(G)$.
    As the paths were edge disjoint in $G'$, they are edge disjoint on all edges of $G$ except for $E_\infty$.
    The other properties hold by definition.
\end{proof}

This resolves the first of the two major differences between $\lenEDP$ and $\lifeTEDSC$.
For the second, observe the following:
When searching paths in $\SE(G, D)$ that cover all of $D$, we already know that for each $(u, v, t) \in D$ the edge $(u_t, v_{t + 1})$ is taken by some path.
To ensure this in a $\lenEDP$ instance, we remove this edge from the input graph and add it back once we know how to route the paths from one demand edge to the next.
The only thing we need to ensure in the $\lenEDP$ instance is that some path ends on the vertex $u_t$ and some path starts on $v_{t + 1}$ so that we can merge these two paths when adding the edge $(u_t, v_{t + 1})$.
Intuitively, given the paths of a $\lifeTEDSC$ solution in $\SE(G, D)$, removing the demand edges splits each path $p$ into a set of smaller paths that connect the endpoints of successive demands covered by $p$.
If we fix for each demand which of the $k$ paths should cover it, we can recover the routing between the demands by solving the corresponding $\lenEDP$ instance.
This process can be seen as a parameterized Turing reduction~\cite{flum2006parameterized}, and we fully define it in \cref{alg:life-fpt-kh}.
We prove that it works correctly in \cref{thm:life-fpt-kh}.
\begin{algorithm}
    \SetupAlgoInputOutput
    \AlgoInput{$\lifeTEDSC$ instance $\mathcal I = (G, D, k, h)$}
    \AlgoOutput{Decision}

    \ForEach{$A\colon D \to [k]$}{
        denote $D_i = \set{d \in D \with A(d) = i}$ for all $i \in [k]$ \;
        \If{$\exists i \in [k], t^* \in [\lifetime]: \cardinality{\set{(u, v, t) \in D_i \with t = t^*}} > 1$}{
            \Continue \tcp*[f]{$i$ serves multiple demands at time step $t^*$}
        }
        \If(\label{line:fpt-kh-check-lifespan}){$\exists i \in [k]: \max\set{t + 1 \with (u, v, t) \in D_i} - \min\set{t \with (u, v, t) \in D_i} > h$}{
            \Continue \tcp*[f]{$i$ lives too long}
        }
        $H \gets \SE_\gapsize(G, D)$ with $\gapsize = (n-1) \cdot k$ (see \cref{def:compressed-tedsc-network}) \;
        \ForEach{$(u, v, t) \in D$}{
            \tcp*[l]{remove demand edges from $H$}
            $H \gets H - \set{(u_t, v_{t + 1})}$ \;
        }
        $H' \gets \PS(H, k, E(H) \setminus E_{\mathrm{move}})$ \;

        \ForEach{$i \in [k]$}{
            $\mathcal S_i \gets \text{empty multiset}$ \;
            enumerate $D_i = \set{(u^1, v^1, t^1), (u^2, v^2, t^2), \dots}$ such that $\forall j: t^j < t^{j + 1}$ \;
            \ForEach{$j \in [\cardinality{D_i} - 1]$}{
                \tcp*[l]{add terminal pairs between demands}
                $\mathcal S_i \gets \mathcal S_i + \set*{\left( v^j_{t^j + 1}, u^{j + 1}_{t^{j + 1}} \right)}$ \;
            }
        }
        $\mathcal S \gets \bigcup_{i = 1}^k \mathcal S_i$ \;
        \If{the $\lenEDP$ instance $(H', \mathcal S, 2h)$ is feasible}{
            \Return YES \;
        }
    }
    \Return NO \;
    \caption{Parameterized Turing Reduction from $\lifeTEDSC$ to $\lenEDP$}
    \label{alg:life-fpt-kh}
\end{algorithm}

An example of the graph $H$ in \cref{alg:life-fpt-kh} and the corresponding $\lifeTEDSC$ and $\lenEDP$ solutions is illustrated in \cref{fig:guessing-graph}.
\begin{figure*}[t]
    \centering
    \begin{subfigure}{0.45\textwidth}
        \centering
        \includegraphics[height=8cm, page = 1]{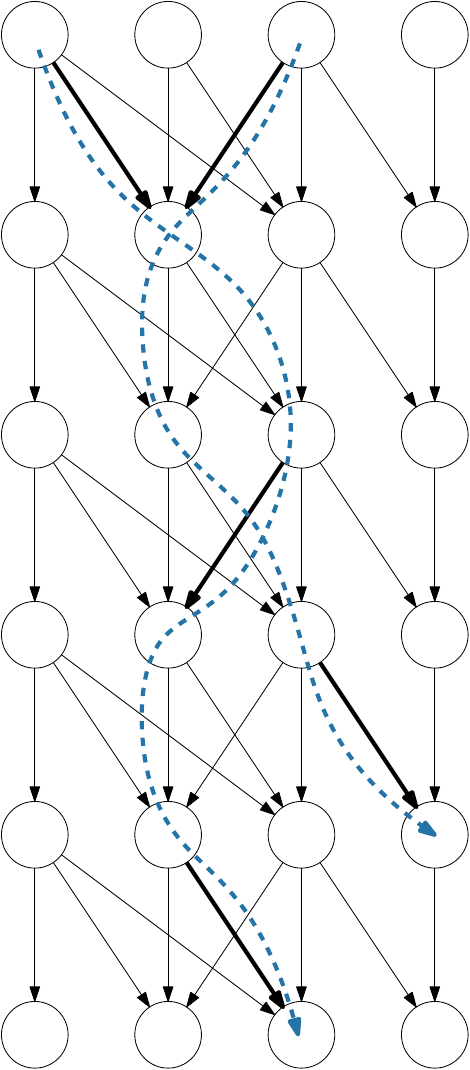}\\[1 ex]
        \caption{Static expansion $\SE(G, D)$ before removing demand edges. Solid, thick lines represent demands. Dashed, blue lines represent a possible $\lifeTEDSC$ solution.}
    \end{subfigure}%
    \qquad
    \begin{subfigure}{0.45\textwidth}
        \centering
        \includegraphics[height=8cm, page = 2]{assets/guessing-reduction.pdf}\\[1 ex]
        \caption{Graph $H$ after removing demand edges. Dashed, blue lines represent a possible $\lenEDP$ solution compatible with the assignment $A$ induced by the schedule.}
    \end{subfigure}%
    \caption{Example instance of \cref{alg:life-fpt-kh} with solutions.}
    \label{fig:guessing-graph}
\end{figure*}

\begin{theorem}
    \label{thm:life-fpt-kh}
    $\lifeTEDSC$ parameterized by $k + h$ is in $\FPT$.
\end{theorem}
\begin{proof}
    See \Cref{alg:life-fpt-kh}.
    Recall that we assume that $\cardinality D \le k \cdot h$ by \cref{thm:D-le-kh}.
    This algorithm terminates in $\FPT$ time because there are $k^{\cardinality D}$ iterations and, as $\cardinality{\mathcal S} \le \cardinality D$, each iteration takes $\FPT$ time by using \cref{thm:len-edp-fpt-kh}.

    Suppose that the algorithm outputs YES in some iteration, and let $A$ be the assignment of this iteration.
    A feasible schedule $S$ can be composed as follows:
    Translate the $\lenEDP$ solution in $H'$ to a set $P$ of paths in $H$ using \cref{thm:ps-paths-exist-iff}.
    Next, expand the edges of the paths in $P$ that are compressed gaps in $\SE_\gapsize(G, D)$ to paths in $\SE(G, D)$ using \cref{thm:tedw-n-times-k}.
    Notably, although the size of $P$ may exceed $k$, every gap is crossed by at most $k$ paths because for every $i \in [k]$ there is at most one terminal pair in $\mathcal S_i$ that connects vertices before and after the gap.
    Thus, the gap size of $\gapsize = (n - 1) \cdot k$ still suffices to apply \cref{thm:tedw-n-times-k} to each gap individually.
    Let $P'$ be the set of paths obtained after expanding edges corresponding to gaps.
    For each $i \in [k]$, construct the temporal walk $p_i \in S$ by concatenating the temporal edges of $D_i$ and the temporal walks corresponding to the paths of $P'$ that connect $\mathcal S_i$.
    The schedule $S$ covers all demands because $D = \bigcup_{i = 1}^k D_i$ and no temporal walk violates the lifespan constraint as we check in Line~\ref*{line:fpt-kh-check-lifespan}.
    The walks of $S$ are TED because each demand edge of $\SE(G, D)$ is used by exactly one path, every moving edge is used by at most one path by \cref{thm:io-paths-iff-compressed}, and the expansion of edges over gaps using \cref{thm:tedw-n-times-k} gives TED temporal walks.

    For the reverse, let $S = \set{p_i}_{i = 1}^k$ be any feasible schedule.
    For every demand $d \in D$, set $A(d) = i$ for the walk $p_i$ which has $d \in \E(p_i)$.
    As every temporal walk contains at most one temporal edge at each time step and as $\lifespan(p_i) \le h$ for all $i \in [k]$, this assignment $A$ passes the two checks at the beginning of the iteration.
    Let $S'$ be the set of temporal walks obtained by splitting each $p \in S$ into the segments between the demands $D_i$ and let $P$ be the set of corresponding paths in $\SE(G, D)$.
    All $p \in P$ satisfy $\len(p) \le h$ because the lifespan of a segment in $S'$ does not exceed the lifespan of the corresponding full walk in $S$, and every walk in $S$ has a lifespan of at most $h$.
    By definition, the paths $P$ connect the terminal pairs $\mathcal S$ and are edge disjoint on $E_{\mathrm{move}}$.
    Thus, mapping $P$ to $H'$ using \cref{thm:ps-paths-exist-iff} gives a set of paths of length at most $2h$ that connect the terminal pairs $\mathcal S$ and are edge disjoint, so the $\lenEDP$ instance is feasible.
\end{proof}

Observe that we did not meaningfully use the length bound of $\lenEDP$ in \cref{alg:life-fpt-kh}.
All paths between vertices of layers $t_1$ and $t_2$ in $\SE(G, D)$ have length $t_2 - t_1$, so finding edge disjoint paths of length at most $h$ is no harder than finding edge disjoint paths of arbitrary length.
However, $\EDP$ parameterized by $k$ is $\W[1]$-hard (even on planar DAGs~\cite{chitnis2023tight}), whereas $\lenEDP$ parameterized by $k + h$ is in $\FPT$.
Generalizing from $\EDP$ to $\lenEDP$ does not complicate the constructed instance, but it allows us to include additional knowledge about the sought paths.

For $\lenTEDSC$, the situation is more complicated.
A temporal walk may wait on a vertex for multiple time steps, which does not count toward its length.
When splitting a temporal walk into multiple terminal pairs that connect successive demands, a segment may span more than $h$ layers while using only a few moving edges.
Instead of limiting the entire walk's lifespan, we need to ensure that the sum of the lengths of the segments does not exceed $h$.
Thus, we need to generalize the $\lenEDP$ problem so that we can specify a length bound for every terminal pair individually.
We refer to this problem as \emph{heterogeneous $\lenEDP$}.
\begin{lemma}
    \label{thm:het-len-edp-fpt}
    There is an algorithm that, given a directed graph $G$ and $k$ terminal pairs with length bounds $\set{(s_i, z_i, h_i)}_{i \in [k]}$ with $h_i \le h$, decides in time $2^{\O(k \cdot h)} \log(n + k \cdot h) (m + k \cdot h) k$ whether there are edge disjoint paths $\set{p_i}_{i \in [k]}$ such that $p_i$ is an $s_i$-$z_i$ path and $\ell(p_i) \le h_i$ for all $i \in [k]$.
\end{lemma}
\begin{proof}
    We reduce this problem to the homogeneous $\lenEDP$ where all length bounds equal $h$.
    For each terminal pair $(s_i, z_i, h_i)$, we extend $G$ by a path of $h - h_i$ new vertices such that one end of this path has an edge to $z_i$ and the other end is the new vertex $z_i'$.
    The new set of terminal pairs is $\set{(s_i, z_i')}_{i = 1}^k$.
    Observe that a $(s_i, z_i)$-path of length at most $h_i$ in the original graph corresponds to a $(s_i, z_i')$-path with length at most $h$ in the new graph.
    Furthermore, as all newly introduced vertices are used by only one path in the solution, the edge disjointness of the original paths is equivalent to that of the new paths.
    The number of terminal pairs $k$ and the bound $h$ are unchanged, and the size of the graph increases by at most $k \cdot h$ vertices and edges.
    Therefore, we can decide the constructed $\lenEDP$ instance using \cref{thm:len-edp-fpt-kh} in the desired time.
\end{proof}

Next, we take care of the fact that in $\lenTEDSC$ the waiting edges do not contribute to the length of the temporal path, yet they are edges in the static expansion and thus contribute to the length of the static paths there.
To reduce $\lenTEDSC$ to $\lenEDP$, we need to create a static graph in which every temporal edge of a walk contributes only some constant length to the corresponding path, regardless of how long it waited before and after taking this temporal edge.
Our construction is similar to that of \textcite{wu2014path} with the difference that we connect the in-vertices to \emph{all} subsequent out-vertices, not just the next one.

\begin{definition}[In-Out expansion]
    \label{def:inout-expansion}
    Given a temporal graph $\G = \langle G, \lambda \rangle$, we define its In-Out expansion $\IO(\G)$ to be the static graph with vertices
    \begin{align*}
        V_t^\vin          & = \set{v_t^\vin \with v \in V(G)},                              \\
        V_t^\vout         & = \set{v_t^\vout \with v \in V(G)},                             \\
        V(\IO(\G))        & = \bigcup_{t \in [\lifetime(\G) + 1]} V_t^\vin \cup V_t^\vout,  \\
        \intertext{and edges}
        E_{\mathrm{wait}} & = \set{(v_t^\vin, v_{t'}^\vout) \with v \in V(G), t \le t'},    \\
        E_{\mathrm{move}} & = \set{(u_t^\vout, v_{t + 1}^\vin) \with (u, v, t) \in \E(\G)}, \\
        E(\IO(\G))        & = E_{\mathrm{wait}} \cup E_{\mathrm{move}}. \qedhere
    \end{align*}
\end{definition}
\begin{figure*}[t]
    \centering
    \includegraphics[height=9cm, page=3]{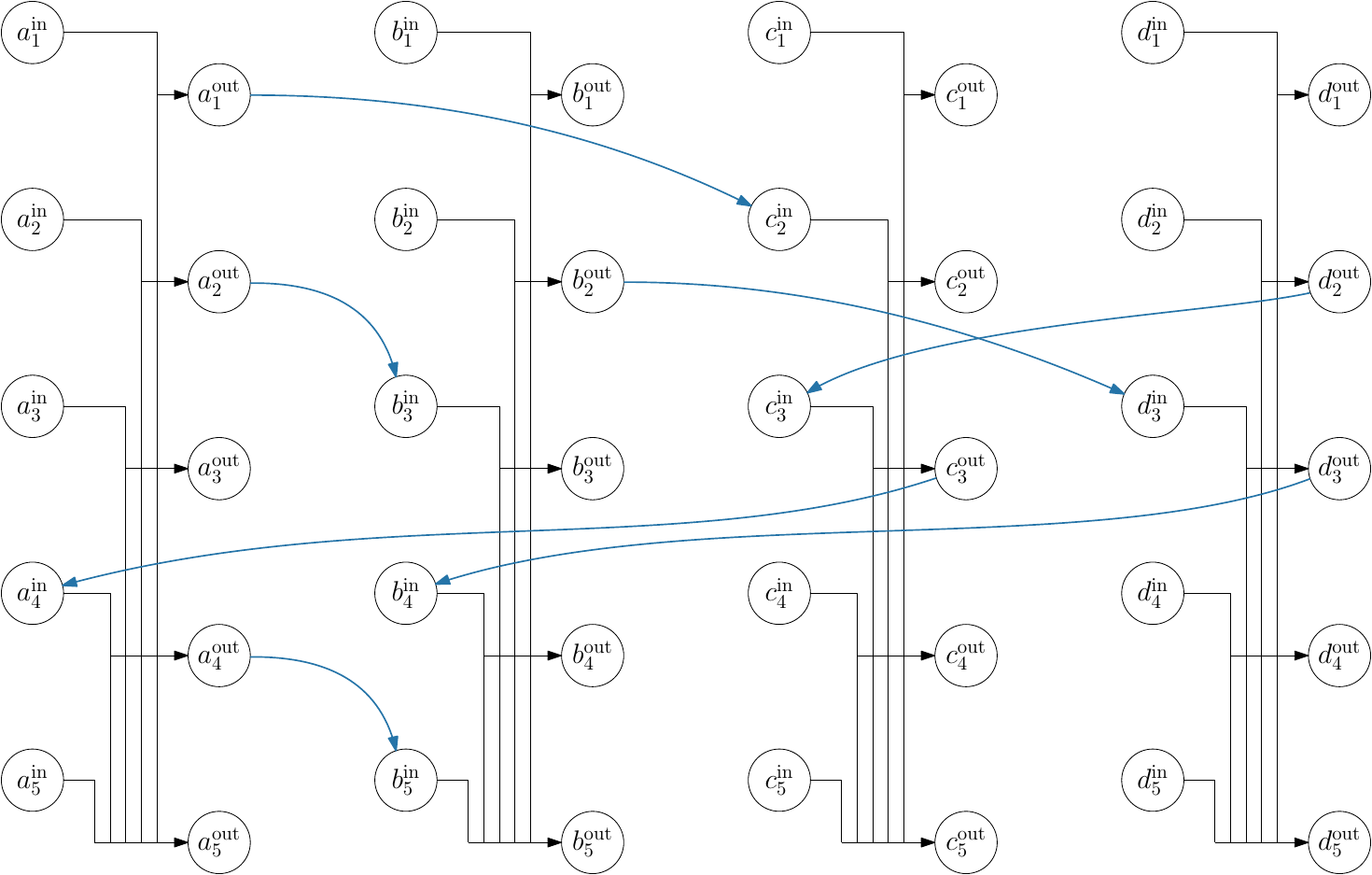}\\[1 ex]
    \caption{
        In-Out expansion $\IO(\G)$ of a temporal graph $\G$.
        Black, axis-aligned lines represent waiting edges.
        Blue, curved lines represent moving edges.
    }
    \label{fig:io-expansion}
\end{figure*}
\Cref{fig:io-expansion} shows an example of \cref{def:inout-expansion}.
Observe that a temporal $u$-$v$-walk $p$ from time step $t_1$ to time step $t_2$ corresponds to a $u_{t_1}^\vout$-$v_{t_2}^\vout$-path $p'$ in $\IO(\G)$.
It holds that $\len(p') = 2 \len(p)$.
Temporal walks covering the demands $D$ are equivalent to paths in $\IO(\G)$ that cover the corresponding moving edges.

Using \cref{thm:tedw-n-times-k} we can compress gaps in $\IO(\G(G, D))$ of size at least $(n-1) \cdot k$ while maintaining feasibility of the instance.
However, in contrast to \cref{def:compressed-tedsc-network}, we cannot simply replace the layers of a gap with dense edges, because we want to preserve the exact distances between vertices across multiple layers.
Instead, we keep $\gapsize$ many layers for the walks to reposition in a gap, and then remove the remaining layers of the gap.
As $\IO(\G)$ already contains waiting edges until all later time steps, we do not have to insert new edges in the compressed expansion.
\begin{definition}
    Given a temporal graph $\G$ and parameter $\gapsize \in \N^+$, the $\gapsize$-compressed In-Out expansion $\IO_\gapsize(\G)$ of $\G$ is defined as the induced subgraph of $\IO(\G)$ containing the vertices $v_t^\vin, v_t^\vout \in V(\IO(\G))$ for which there exists a relevant time step $t^*$ with $t \in [t^*, t^* + 1 + \gapsize]$.
\end{definition}
Similar to \cref{thm:tedsc-poly-n-d}, for suitably large gap sizes $\gapsize$, compression does not change the feasibility of instances.
\begin{lemma}
    \label{thm:io-paths-iff-compressed}
    Let $(G, D, k, h)$ be a $\lenTEDSC$ instance and let $\gapsize \ge (n-1) \cdot k$.
    There are $k$ paths in $\IO(G, D)$ of length at most $2h$ that cover $D$ and are edge disjoint on $E_{\mathrm{move}}$ if and only if there are such paths in $\IO_\gapsize(G, D)$.
\end{lemma}
\begin{proof}
    The reverse direction holds by definition.
    For the other, let $\set{p_i}_{i = 1}^k$ be the set of paths in $\IO(G, D)$.
    For every gap $[t_1, t_2 - 1]$ of size at least $\gapsize$, we replace the edges of each path $p_i$ during the gap with a walk that uses vertices that are also present in $\IO_\gapsize(G, D)$.
    Let $s_i$ and $z_i$ be the vertices that $p_i$ is on at time steps $t_1$ and $t_2$, respectively.
    \Cref{thm:tedw-n-times-k} with $F = E(G)$ gives a set $S$ of shortest, TED temporal paths connecting $\set{(s_i, z_i)}_{i = 1}^k$.
    Mapping the paths in $S$ to $\IO(G, D)$ and using them to replace the subpath of each $p_i$ in the gap, we get a new set of paths that covers $D$.
    As $S$ contains only shortest paths, the length of each $p_i$ does not increase.
    Moreover, as the paths in $S$ take at most $\gapsize$ time steps, all $\set{p_i}_{i = 1}^k$ now use vertices from $\IO_\gapsize(G, D)$ during the gap $[t_1, t_2 - 1]$.
    As the paths of $S$ are TED, the paths in $\IO_\gapsize(G, D)$ remain edge disjoint on $E_{\mathrm{move}}$.
    Repeating this replacement for all gaps yields the desired set of paths.
\end{proof}

\begin{theorem}
    \label{thm:len-fpt-kh}
    $\lenTEDSC$ parameterized by $k + h$ is in $\FPT$.
\end{theorem}
\begin{proof}
    The algorithm is a modification of \Cref{alg:life-fpt-kh}.
    In addition to iterating over all assignments $A\colon D \to [k]$, there is a second nested loop iterating over all assignments $A'\colon D \to [h]$.
    The value $A'(d)$ fixes the number of remaining steps that the walk covering $d$ should have right before traversing the temporal edge of $d$.
    If there are demands $d = (u, v, t), d' = (u', v', t') \in D$ with $A(d) = A(d')$ and $t < t'$, but $A'(d') > A'(d)$, we continue to the next iteration.
    Such an assignment would mean that the number of remaining steps of walk $i$ has increased from when demand $d$ was covered to later when $d'$ was covered.
    In turn, we do not perform the check of Line~\ref*{line:fpt-kh-check-lifespan}, so a walk can be assigned to demands that are more than $h$ time steps apart.

    Next, instead of assigning $H$ to be $\SE_\gapsize(G, D)$ we use $\IO_\gapsize(G, D)$.
    Similar to before, the moving edges corresponding to demands are removed from $H$, and we set $H' \gets \PS(H, k, E(H) \setminus E_{\mathrm{move}})$.
    When constructing the terminal pairs $\mathcal S$, the terminal pair from $d = (u, v, t)$ to $d' = (u', v', t')$ should be from $v_{t+1}^\vin$ to $(u')_{t'}^\vin$ and have a length bound of $4(A'(d) - A'(d') - 1)$.
    The factor of $4$ corresponds to the doubling of length from $\G$ to $\IO_\gapsize(G, D)$ and another doubling due to the $\PS$ transformation.
    Finally, instead of solving a $\lenEDP$ instance, we solve heterogeneous $\lenEDP$ in $H'$ with the terminal pairs $\mathcal S$ using \Cref{thm:het-len-edp-fpt}.

    The proof for correctness is analogous to that of \Cref{thm:life-fpt-kh}, with the additional use of \cref{thm:io-paths-iff-compressed}.
    For the running time, there are $k^{\cardinality D}$ and $h^{\cardinality D}$ possible values for the assignments $A$ and $A'$, respectively.
    The heterogeneous $\lenEDP$ instance is a graph of polynomial size, the number of terminal pairs is at most $\cardinality D$, and each length bound is at most $4h$.
    By \Cref{thm:D-le-kh,thm:het-len-edp-fpt} we get the desired running time.
\end{proof}

Note that modifying the algorithms from \cref{thm:life-fpt-kh,thm:len-fpt-kh} yields $\XP$ algorithms for constrained $\TEDSC$ parameterized by $\cardinality D$.
We omit specifying the details for these two algorithms and instead move on to the next section in which we present an $\XP$ algorithm for the strictly smaller parameter $k$ instead.
A sketch for how these algorithms for the parameter $\cardinality D$ work is as follows:
For $\lifeTEDSC$, we reduce to (unbounded) $\EDP$ on DAGs parameterized by $k$ and use the $\XP$ algorithm by \textcite{fortune1980directed}.
For $\lenTEDSC$, we solve the $\lenEDP$ instance using a modified version of the $\XP$ algorithm that additionally incorporates length bounds.
As guessing $A$ and $A'$ takes $k^{\O(\cardinality D)}$ and $h^{\O(\cardinality D)}$ time respectively, the overall running time is $\XP$ in $\cardinality D$ too (recall that in $\lenTEDSC$, \cref{thm:length-h-le-nD} allows us to assume that $h$ is polynomial in the input size).

\subsection{XP Algorithm using State Graphs}
\label{sec:k-xp}
In \Cref{sec:fpt-kh} we saw that both $\lenTEDSC$ and $\lifeTEDSC$ are in $\FPT$ when parameterized by $k + h$.
Recalling that we assume $k < \cardinality D \le k \cdot h$, \cref{thm:hard-by-D-planar-dag,thm:life-fpt-kh,thm:len-fpt-kh} state that the gap for fixed parameter tractability is between $\cardinality D$ and $k + h$.
Assuming $\FPT \ne \W[1]$, \Cref{thm:hard-by-D-planar-dag} implies that there is no $\FPT$ algorithm for constrained $\TEDSC$ parameterized by $\cardinality D$.
However, it might still be possible to construct $\XP$ algorithms, as those would only be ruled out by a corresponding $\paraNP$-hardness result.
To this end, we know from \Cref{thm:h5-hard} that parameterizing by $h$ cannot yield an $\XP$ algorithm unless $\P = \NP$.

In this section, we show that both constrained $\TEDSC$ variants parameterized by $k$ are in $\XP$.
By our usual assumption $k < \cardinality D$, this algorithm also works for parameterizing by $\cardinality D$.
The main concept of the algorithm is the \emph{state graph} $\St(\mathcal I)$ of an instance $\mathcal I$.
Every vertex of $\St(\mathcal I)$ corresponds to a configuration of walks in $G$ at a time step $t$.
Edges in the state graph correspond to valid state transitions between two successive time steps.
As before, we compress the state graph to only a bounded number of steps around relevant time steps.

First, we formally define the state graph.
After the definition, we show that traversing this graph corresponds to constructing schedules, which also gives some intuition for the definition.
\begin{definition}
    \label{def:state-graph}
    Let $\mathcal I = (G, D, k, h)$ be a $\lenTEDSC$ or $\lifeTEDSC$ instance.
    The \emph{state graph} $\St(\mathcal I)$ is a directed graph where $V(\St(\mathcal I))$ is the set of states of $\mathcal I$.
    A state is a tuple $(t, U, A, p, c)$ where $t \in [\lifetime(D) + 1]$ is a time step, $U \subseteq [k]$ is a set of indices of so-far unused walks, $A \subseteq [k] \setminus U$ is a set of indices of active walks, $p\colon A \to V(G)$ is a position profile of the active walks, and $c\colon A \to [0, h]$ is a capacity profile of the active walks.
    We say that a state is an \emph{initial state} if $t = 1$, and a \emph{terminal state} if $t = \lifetime(D) + 1$.
    The state graph contains an edge from a state $(t_1, U_1, A_1, p_1, c_1)$ to a state $(t_2, U_2, A_2, p_2, c_2)$ if and only if
    all of the following hold:
    \begin{enumerate}
        \item\label{item:state-time-delta} $t_2 = t_1 + 1$.
        \item $U_2 \subseteq U_1$.
        \item $A_2 \subseteq A_1 \cup U_1$.
        \item
              For every $i \in A_1 \cap A_2$ it holds that
              \begin{enumerate}
                  \item $p_1(i) = p_2(i)$ or $(p_1(i), p_2(i)) \in E(G)$.
                  \item If $\mathcal I$ is a $\lenTEDSC$ instance then $c_2(i) = c_1(i) - d_G(p_1(i), p_2(i))$.
                  \item If $\mathcal I$ is a $\lifeTEDSC$ instance then $c_2(i) = c_1(i) - (t_2 - t_1)$.
              \end{enumerate}
        \item
              Let $\E(p_1, p_2) = \set{(p_1(i), p_2(i), t_1) \with i \in A_1 \cap A_2 \land p_1(i) \ne p_2(i)}$ denote the multiset of temporal edges used in the transition.
              \begin{enumerate}
                  \item Every temporal edge occurs at most once in $\E(p_1, p_2)$.
                  \item For every demand $d = (u, v, t) \in D$ with $t = t_1$, $d \in \E(p_1, p_2)$. \qedhere
              \end{enumerate}
    \end{enumerate}
\end{definition}
Intuitively, the first condition ensures that state transitions are only forwards in time, the second and third ensure that each walk goes from \emph{unused} over \emph{active} to \emph{finished}, the fourth ensures that the capacity of each walk decreases correctly, and the fifth ensures that the walks are temporally edge disjoint and cover all demands.

\begin{lemma}
    \label{thm:st-path-iff}
    Every schedule for a constrained $\TEDSC$ instance $\mathcal I$ corresponds to a path in $\St(\mathcal I)$ from some initial state to some terminal state, and vice versa.
\end{lemma}
\begin{proof}
    Suppose that there is such a path $P$ in $\St(\mathcal I)$.
    Then this induces a schedule $S$ where the temporal edges of each walk $i \in [k]$ are exactly those of the form $(p_1(i), p_2(i), t)$ when $i \in A_1 \cap A_2$ and $p_1(i) \ne p_2(i)$.
    Throughout the states of $P$, each walk $i$ first is in the unused set $U$, then at some point transitions to the active set $A$, and finally is removed from $A$ and never occurs in either of these sets again.
    Moreover, as $c(i)$ is initially equal to $h$ and then decreases according to the rules of $\lenTEDSC$ / $\lifeTEDSC$, the length / lifespan of walk $i$ is at most $h$.
    Finally, the last set of conditions ensures that the walks of $S$ are TED and cover all demands $D$.
    For the other direction, every schedule $S$ induces a corresponding path through $\St(\mathcal I)$ with the individual states matching the positions and remaining capacities of the walks at each time step.
\end{proof}

Thus, finding such a path is the same as solving $\mathcal I$.
However, constructing $\St(\mathcal I)$ naively yields too many vertices to still find paths in time $\cardinality{\mathcal I}^{f(k)}$.
Thus, we define the compressed state graph $\St_\gapsize(\mathcal I)$ for a given gap size $\gapsize \in \N$.
\begin{definition}
    \label{def:compressed-state-graph}
    The \emph{compressed state graph} $\St_\gapsize(\mathcal I)$ of a constrained $\TEDSC$ instance $\mathcal I$ is obtained from $\St(\mathcal I)$ as follows:
    \begin{enumerate}
        \item
              Let $[t_1, t_2 - 1]$ be a gap of size at least $\gapsize$.
              We omit all states at time steps $t$ with $t \in [t_1 + \gapsize, t_2 - 1]$.
        \item
              We relax Condition~\ref*{item:state-time-delta} of \cref{def:state-graph} to allow transitions from a time step $t$ to $t'$ if $t'$ is the smallest time step after $t$ at which a state exists.
        \item
              If $\mathcal I$ is a $\lifeTEDSC$ instance, then we omit states $(t, U, A, p, c)$ unless they fulfill the following property:
              For all $i \in A$ we have $c(i) \in \set{h - (t - t_0) \with t_0 \in R}$, where $R$ denotes the set of relevant time steps (see \cref{def:relevant-time}). \qedhere
    \end{enumerate}
\end{definition}
The first step ensures that the number of time steps is $\O(\cardinality D \cdot \gapsize)$, which is polynomial for $\gapsize = (n - 1) \cdot k$.
The second step adds state transitions to skip the states that were removed in the first step.
The third step ensures that the number of possible values $c$ is bounded by $\cardinality D$, not only by $h$.
Notably, for $\lenTEDSC$ the upper bound of $h$ suffices because $h$ is polynomially bounded by \cref{thm:length-h-le-nD}.

\begin{theorem}
    \label{thm:life-len-k-xp}
    There is an algorithm that, given a $\lenTEDSC$ or $\lifeTEDSC$ instance $\mathcal I$, decides whether a feasible schedule exists in time $\cardinality{\mathcal I}^{\O(k)}$.
\end{theorem}
\begin{proof}
    Constructing the compressed state graph $\St_\gapsize(\mathcal I)$ with $\gapsize = (n - 1) \cdot k$ for an instance $\mathcal I$ yields a graph with the number of vertices $\cardinality{V(\St_\gapsize(\mathcal I))}$ bounded by
    \[
        \begin{bracealign}
            \underbrace{\cardinality D \cdot n \cdot k}_t
            \cdot
            \underbrace{2^k}_U
            \cdot
            \underbrace{2^k}_A
            \cdot
            \underbrace{n^k}_p
            \cdot
            \underbrace{C^k}_c
        \end{bracealign} \le \cardinality{\mathcal I}^{\O(k)}
    \]
    where $C = h + 1 \in \O(n \cdot \cardinality D)$ for $\lenTEDSC$ (see \cref{thm:length-h-le-nD}) and $C \le \cardinality D$ for $\lifeTEDSC$.
    By \cref{thm:st-path-iff}, a schedule exists if and only if $\St(\mathcal I)$ contains a path from an initial state to a terminal state.
    To prove that such a path is also witnessed by $\St_\gapsize(\mathcal I)$, observe that by \cref{thm:tedw-n-times-k} we can assume that no temporal walk of a schedule moves at a time step $t \in [t_1 + \gapsize, t_2 - 1]$ where $[t_1, t_2 - 1]$ is a gap.
    Furthermore, we can assume that the first temporal edge of each walk is at a relevant time step, so the remaining capacity of each walk at each time step is contained in the set of allowed values for $c$.
    For the other direction, a state transition that skips multiple time steps can be replicated by multiple transitions in $\St(\mathcal I)$ in which no temporal walk moves.
    As the number of edges in $\St_\gapsize(\mathcal I)$ is at most $(\cardinality{\mathcal I}^{\O(k)})^2$, the existence of a path from an initial state to a terminal state can be checked in the desired time.
\end{proof}

\subsection{Polynomial Time Algorithm for Fixed Bidirected Stars}
\label{sec:star-algorithm}
Our results in \cref{sec:fpt-kh,sec:k-xp} suggest that $\lenTEDSC$ and $\lifeTEDSC$ are equally complex.
In particular, we have not observed any differences in the parameterized complexity of these two problems with respect to the inherent parameters $k$, $h$, and $\cardinality D$.
We did, however, notice that constructing $\FPT$ algorithms parameterized by $k + h$ requires more work for $\lenTEDSC$ than for $\lifeTEDSC$.
Furthermore, the strongest hardness result we saw so far, \Cref{thm:length-w-hard-on-p3}, holds only for $\lenTEDSC$ and there is no analog of this reduction for $\lifeTEDSC$.

We confirm the intuition that $\lenTEDSC$ is computationally more difficult than $\lifeTEDSC$.
In particular, we construct an algorithm that forbids the existence of a hardness result like \Cref{thm:length-w-hard-on-p3} for $\lifeTEDSC$, unless $\P = \NP$.
To contrast \Cref{thm:length-w-hard-on-p3}, the goal is to construct an $\FPT$ algorithm for the parameter $k$ on the bidirected $P_3$ graph.
In this section, we present an $\XP$ algorithm for $\lifeTEDSC$ parameterized by $n$ on bidirected star graphs.
We thus exceed our goal in two ways:
First, the algorithm terminates not only in $\FPT$ time by $k$, but in polynomial time on any \emph{fixed} bidirected star graph.
Second, it works not only for $P_3$ but also for the entire class of bidirected star graphs.
For the rest of this section, assume that $G$ is a bidirected star graph as defined below.
\begin{definition}[Bidirected star graph]
    A \emph{bidirected star graph} is a graph $G$ with a \emph{center vertex} $c \in V(G)$ such that $E(G) = \bigcup_{v \in V(G) \setminus \set c} \set{(c, v), (v, c)}$.
    The vertices $V(G) \setminus \set c$ are called the \emph{leaves} of $G$.
\end{definition}

\Cref{fig:star-graph} shows an example of a bidirected star graph.
\begin{figure}[t]
    \centering
    \includegraphics[height=3.5cm]{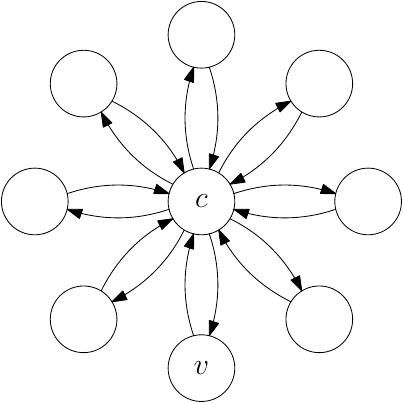}
    \caption{
        Bidirected star graph with $9$ vertices.
        The vertex~$c$ is the center vertex, the vertex~$v$ is a leaf.
    }
    \label{fig:star-graph}
\end{figure}
Our algorithm is based on the state graph algorithm from \cref{sec:k-xp}, which is $\XP$ by $k$.
We show how to limit the set of states that we need to explore.
Specifically, we limit the possible choices for the sets of unused ($U$) and active ($A$) temporal walks in each state.
Our approach works in three steps:
\begin{enumerate}
    \item
          We define the \emph{concentration} property and present the $\Concentrate$ procedure, which transforms arbitrary schedules into concentrated ones.
          Intuitively, if a schedule is concentrated, no walk idles on a leaf for a long time.
    \item
          We use the concentration property in the procedure $\Sparsify$, which also ensures that no excessive number of temporal walks are waiting at the center vertex.
          Combined with concentration, this means that the temporal walks do not idle for long times anywhere on $G$.
          Thus, the temporal walks need to frequently traverse temporal edges of $\G$ but also stay temporally edge disjoint.
          Consequently, the number of temporal walks active at any given time step is bounded by a function of the number of edges of $G$.
    \item Independently of the previous two steps, we argue in \cref{thm:st-U-prefix} that only linearly many choices for $U$ need to be considered.
\end{enumerate}
Notably, although we describe these modifications as procedures, they are not part of the final algorithm, so they do not have to terminate in the desired running time.
Instead, we show that they terminate after a finite number of steps, thereby proving the existence of a schedule with the desired properties.
By this exchange argument, it suffices to explore only the part of the state graph corresponding to schedules that fulfill these properties.

We begin with the formal definition of concentration.
Recall that for a temporal walk $p$, $\firstt(p)$ denotes the time step of the first temporal edge of $p$.
Thus, $\firstt(p) + h$ is the first time step at which $p$ cannot traverse temporal edges anymore, because then the lifespan of $p$ would exceed $h$.

\begin{definition}[Concentrated schedule]
    \label{def:concentrated-schedule}
    A schedule $S$ is \emph{concentrated until time step $t$} if the following holds for all $p \in S$:
    If $p$ enters a leaf $v$ using a temporal edge $(c, v, t^\vout)$ with $t^\vout \le t$, then it returns to $c$ using $(v, c, t^\vin)$ at some time step $t^\vin \in [t^\vout + 1, t^\vout + 2]$ or it holds that $t^\vout + 2 \ge \min(\lifetime + 1, \firstt(p) + h)$.
    If $S$ is concentrated until time step $t = \lifetime$, then $S$ is \emph{concentrated}.
\end{definition}

We say that a temporal walk $p \in S$ takes a \emph{violating temporal edge} $(c, v, t^\vout)$ if the above conditions are violated for this choice of variables.
A concentrated schedule can be obtained from an arbitrary schedule $S$ using the $\Concentrate$ procedure.
For the rest of this section, it helps to view the temporal walks $p \in S$ as sets of temporal edges rather than a sequence of visited vertices and time steps.
All procedures add or remove temporal edges from $\E(p)$ while ensuring that this set still induces a connected sequence of temporal edges.

Intuitively, to establish concentration, we do the following:
Whenever a temporal walk $p \in S$ takes a violating edge $(c, v, t)$, we want to fix this by packing more temporal edges into $p$.
If we add the temporal edge $(v, c, t + 1)$ to $p$, we resolve the violation.
However, if $p$ would return to the center $c$ at a later time step $\tnext$, then adding $(v, c, t + 1)$ disconnects $p$.
To fix this, we need to add one more edge such that $p$ comes to $v$ again at time step $\tnext - 1$.
This ensures that $p$ is back in time to take the temporal edge $(v, c, \tnext)$.

Both temporal edges $(v, c, t + 1)$ and $(c, v, \tnext - 1)$ might already be used by some other temporal walk $q \in S$, so simply adding these temporal edges to $\E(p)$ would result in $p$ and $q$ not being TED anymore.
We define two auxiliary procedures $\ShiftBackwards$ and $\ShiftForwards$ that allocate these temporal edges for $p$, that is, they modify the schedule $S$ such that the desired temporal edges are free.
The objective of both of these procedures is to create tight pairs of temporal edges $(c, v, t)$ and $(v, c, t+1)$ that can be used by the same temporal walk.
The precondition for both procedures is that a given walk $p$ already occupies one of these two temporal edges, and the result is that the other is no longer occupied by any walk in $S$.
If the desired temporal edge is the first temporal edge of another walk $q$, then it is safe to remove it from $q$.
Otherwise, we can recursively allocate a new temporal edge for $q$, such that the desired temporal edge for $p$ is free.
Both $\ShiftBackwards$ and $\ShiftForwards$ work by the same principle, but $\ShiftBackwards$ is slightly simpler because we can use the fact that the schedule $S$ is already concentrated until $t - 1$ which lets us exactly deduce the time step of the previous temporal edge of $q$.

Although it is easy to see that both $\ShiftBackwards$ and $\ShiftForwards$ terminate after a finite number of recursive invocations, it is not obvious why $\Concentrate$ does not iterate indefinitely.
In particular, each iteration fixes one walk at one time step but may create additional violations, possibly at earlier time steps.
We prove that $\Concentrate$ terminates after the presentation of the pseudocode.
\begin{procedure}
    \SetupAlgoPrePost
    \AlgoPost{$S$ is concentrated}

    \While{$S$ is not concentrated}{
        let $t \in [\lifetime]$ be minimal such that $S$ is not concentrated until time step $t$ \;
        let $p \in S$ be the temporal walk taking the violating temporal edge $(c, v, t)$ \;

        call \ShiftBackwards($S$, $p$, $v$, $t + 1$) \;
        add $(v, c, t + 1)$ to $p$ \; \label{line:concentrate-add-after-shift-back}
        \If{there is $(v, c, t') \in \E(p)$ with $t + 2 \le t'$}{
            let $\tnext = \min\set{t' \in [t + 2, \lifetime] \with (v, c, t') \in \E(p)}$ \;
            call \ShiftForwards($S$, $p$, $v$, $\tnext - 1$) \;
            add $(c, v, \tnext - 1)$ to $p$ \; \label{line:concentrate-add-after-shift-forw}
        }
    }

    \caption{Concentrate($S$)}
\end{procedure}

\begin{procedure}
    \SetupAlgoPrePost
    \AlgoPrec{$S$ is concentrated until $t - 2$}
    \AlgoPrec{$(c, v, t - 1) \in \E(p) \land (v, c, t) \notin \E(p)$}
    \AlgoPost{$(v, c, t) \notin \E(S)$}

    \If{$(v, c, t) \notin \E(S)$}{
        \Return \; \label{line:line:shiftbackwards-empty-found}
    }
    let $q \in S \setminus \set p$ with $(v, c, t) \in \E(q)$ \;
    \If{$(v, c, t)$ is the first temporal edge of $q$}{
        remove $(v, c, t)$ from $q$ \;
        \Return \;
    }
    by preconditions we have $(c, v, t - 2) \in \E(q)$ \;
    call \ShiftBackwards($S$, $q$, $v$, $t-1$) \;
    remove $(v, c, t)$ from $q$ \;
    add $(v, c, t-1)$ to $q$ \; \label{line:shiftbackwards-add-after-recurse}

    \caption{ShiftBackwards($S$, $p$, $v$, $t$)}
\end{procedure}

\begin{procedure}
    \SetupAlgoPrePost
    \AlgoPrec{$(c, v, t) \notin \E(p) \land (v, c, t + 1) \in \E(p)$}
    \AlgoPost{$(c, v, t) \notin \E(S)$}

    \If{$(c, v, t) \notin \E(S)$}{
        \Return \;
    }
    let $q \in S \setminus \set p$ with $(c, v, t) \in \E(q)$ \;
    \If{$(c, v, t)$ is the last temporal edge of $q$}{
        remove $(c, v, t)$ from $q$ \;
        \Return \;
    }
    let $t' = \min\set{t' \in [\lifetime] \with (v, c, t') \in \E(q) \land t < t'}$ \;
    \tcp*[l]{we know that $t + 1 < t'$, because $(v, c, t + 1) \in \E(p)$}
    call \ShiftForwards($S$, $q$, $v$, $t' - 1$) \;
    remove $(c, v, t)$ from $q$ \;
    add $(c, v, t' - 1)$ to $q$ \; \label{line:shiftforwards-add-after-recurse}

    \caption{ShiftForwards($S$, $p$, $v$, $t$)}
\end{procedure}

We begin by arguing that the modifications made by the presented procedures $\Concentrate$, $\ShiftBackwards$, and $\ShiftForwards$ are sound.

\begin{lemma}
    \label{thm:concentrate-sound}
    Let $S$ be a feasible schedule and $p \in S$.
    All modifications made to $p$ in $\Concentrate$ ensure that $\E(p)$ still induces a connected sequence of temporal edges in $\G$.
\end{lemma}
\begin{proof}
    Whenever a walk $p$ is modified, then one of the following holds:
    \begin{enumerate}
        \item A temporal edge $(u, v, t)$ is removed and is the first or last temporal edge of $p$.
        \item A temporal edge $(u, v, t)$ is added as the first or last temporal edge of $p$, and the endpoints of the temporal edges match the rest of $p$.
        \item
              A temporal edge $(u, v, t)$ is replaced by another temporal edge $(u, v, t')$.
              Let $t_<$ and $t_>$ be the time steps of the previous and next temporal edges of $p$ relative to $(u, v, t)$.
              Then it holds that $t' \in [t_< + 1, t_> - 1]$, so $p$ stays connected.
        \item
              The temporal edges $(v, c, t + 1)$ and $(c, v, \tnext - 1)$ are added in $\Concentrate$.
              Before this, $p$ does not contain any temporal edges during $[t+1, \tnext - 1]$, and as $(c, v, t)$ is a violating temporal edge, we have $\tnext > t + 2$.
              Thus, adding the two temporal edges leaves $p$ connected.
    \end{enumerate}
    The temporal edges added are on static edges of the form $(c, v)$ or $(v, c)$, which are both in $E(G)$.
    In most cases, the time step of an added temporal edge is between the time steps of two temporal edges that existed before the addition.
    The only place where this does not hold is in Line~\ref*{line:concentrate-add-after-shift-back} of $\Concentrate$.
    However, in this case, we know that $t + 1 \le \lifetime$ because the chosen time step $t$ satisfies $t + 2 < \lifetime + 1$ by \cref{def:concentrated-schedule}.
    Therefore, all added temporal edges $(u, v, t)$ satisfy $(u, v) \in E(G)$ and $t \in [\lifetime]$.
\end{proof}

Next, we prove that applying $\Concentrate$ to a feasible schedule returns another feasible schedule.
The following lemmas argue about a single iteration of $\Concentrate$.
Once we establish that $\Concentrate$ terminates in \cref{thm:concentrate-potential-inc}, overall correctness follows by induction.

\begin{lemma}
    \label{thm:concentrate-monotone}
    Let $S$ be a feasible schedule and $S^*$ be the result of running one iteration of $\Concentrate$ on $S$.
    Then $\E(S) \subseteq \E(S^*)$ and $S^*$ is TED.
\end{lemma}
\begin{proof}
    Every time that a temporal edge is removed from $S$ in $\ShiftBackwards$ or $\ShiftForwards$, the calling procedure adds the same temporal edge to another walk right afterward.
    Moreover, whenever a temporal edge is added to a walk, then this is preceded by a call to $\ShiftBackwards$ or $\ShiftForwards$ to ensure that the added edge is not occupied by another walk already.
\end{proof}

\begin{lemma}
    \label{thm:concentrate-lifespan-constraint}
    Let $S$ and $S^*$ be defined as in \cref{thm:concentrate-monotone}.
    Then all temporal walks $p^* \in S^*$ satisfy $\lifespan(p^*) \le h$.
\end{lemma}
\begin{proof}
    We check all places where temporal edges are added to temporal walks in $\Concentrate$.
    Initially, all $p \in S$ satisfy $\lifespan(p) \le h$.
    In Line~\ref*{line:concentrate-add-after-shift-back} of $\Concentrate$, we know that $p$ violates concentration, so we can add a temporal edge at time step $t+1$ without exceeding the lifespan bound of $h$.
    In Line~\ref*{line:concentrate-add-after-shift-forw} of $\Concentrate$, $p$ contains a temporal edge at the time step $\tnext$, so adding the temporal edge at $\tnext - 1$ does not increase the lifespan of $p$.
    Next, in Line~\ref*{line:shiftbackwards-add-after-recurse} of $\ShiftBackwards$, we know that $q$ previously contained temporal edges at the time steps $t - 2$ and $t$, so adding one at time step $t - 1$ does not increase the lifespan of $q$.
    Analogously, in Line~\ref*{line:shiftforwards-add-after-recurse} of $\ShiftForwards$, $q$ previously contained temporal edges at time step $t < t' - 1$ and $t'$, so adding one at time step $t' -  1$ does not increase the lifespan of $q$.
\end{proof}

On the one hand, \cref{thm:concentrate-monotone} guarantees that if $S$ covers $D$ then after running $\Concentrate(S)$, $S$ still covers $D$.
Combined with \cref{thm:concentrate-lifespan-constraint}, this means that the modified schedule is still a feasible solution.
On the other hand, \cref{thm:concentrate-monotone} is the basis for the termination proof of $\Concentrate$.
If every iteration would yield $\E(S) \subsetneq \E(S^*)$, then we could conclude that $\Concentrate$ terminates because $\E(S^*)$ cannot exceed $\E(\G)$, which is finite.
However, for some iterations, we have $\E(S) = \E(S^*)$.
This occurs only if the start of some walk is delayed, which can happen only a finite number of times.
\begin{definition}
    Let $S$ be a feasible schedule.
    Order the temporal walks $p \in S$ by $\firstt(p)$ as $p_1, p_2, \dots, p_{\cardinality S}$.
    Let $\firstt_i(S) = \firstt(p_i)$ if $i \in [\cardinality{S}]$ and $\firstt_i(S) = \lifetime + 1$ otherwise.
    We define the \emph{potential} of $S$ as the tuple
    \[
        \Phi(S) = \langle \cardinality{\E(S)}, \firstt_1(S), \firstt_2(S), \dots, \firstt_k(S) \rangle.
    \]
    Given two schedules $S$ and $S'$, we compare $\Phi(S)$ and $\Phi(S')$ lexicographically with respect to $\cardinality{\E(S)}$, and element-wise with respect to $\firstt_1(S), \firstt_2(S), \dots, \firstt_k(S)$.
    Formally, we define $\Phi(S) \le \Phi(S')$ if $\cardinality{\E(S)} < \cardinality{\E(S')}$, or $\cardinality{\E(S)} = \cardinality{\E(S')}$ and $\firstt_i(S) \le \firstt_i(S')$ for all $i \in [k]$.
\end{definition}

\begin{lemma}
    \label{thm:concentrate-potential-inc}
    Let $S$ and $S^*$ be defined as in \cref{thm:concentrate-monotone}.
    Then $\Phi(S) < \Phi(S^*)$.
\end{lemma}
\begin{proof}
    By \cref{thm:concentrate-monotone}, we have $\E(S) \subseteq \E(S^*)$.
    If $\E(S) \subsetneq \E(S^*)$, the potential increases, so assume $\E(S) = \E(S^*)$ going forward.
    In this case, all added edges were previously removed from another temporal walk.
    As no edge is removed from $p$ in $\Concentrate$, the first edge of some other walk $q$ must have been removed in the deepest recursive call of $\ShiftBackwards$, so the start of $q$ is delayed.
    Moreover, none of the changes made in $\Concentrate$ result in a walk starting earlier in $S^*$ than in $S$.
\end{proof}

\begin{lemma}
    \label{thm:concentrate-terminates}
    $\Concentrate$ terminates.
\end{lemma}
\begin{proof}
    By \cref{thm:concentrate-potential-inc}, the potential $\Phi(S)$ of the schedule $S$ increases in every iteration, but it cannot exceed $\langle \cardinality{\E(\G)}, \lifetime + 1, \dots, \lifetime + 1 \rangle$.
    Therefore, there can only be finitely many iterations.
\end{proof}

We have thus established that if there is a feasible schedule, there is also a concentrated one.
This already allows us to limit the set of states considered in the state graph when running \cref{thm:life-len-k-xp} on bidirected stars:
As walks need to return to the center after a constant number of time steps, only a linear number of walks can be on leaves at any given time step.
However, the number of walks waiting at the center vertex remains unbounded, and we still need to store the remaining capacity of all these walks.
This amounts to $\cardinality{D}^k$ possible states, which (for constant $n$) is more than the $\cardinality{\mathcal I}^{f(n)}$ we are aiming for.

Thus, the next step is to reduce the number of walks on the center vertex $c$.
We define the $\Sparsify$ procedure, which produces a schedule in which only a bounded number of walks are active at each time step.
Intuitively, to use the temporal walks as effectively as possible, each one should be traversing $\G$ as long as possible.
When introducing $p \in S$ at the time step $\firstt(p)$, then it should stay in the set of active walks $A$ of the state graph for as many time steps as possible.
As the lifespan of $p$ must not exceed $h$, the last time step at which $p$ can be in the set $A$ of a state is $\firstt(p) + h$.
This motivates the following definition.
\begin{definition}
    Let $S$ be a feasible schedule and $t \in [\lifetime]$.
    The subset of temporal walks of $S$ that should be considered active at time step $t$ is defined as
    \[
        S_t = \set{p \in S \with t \in [\firstt(p), \firstt(p) + h]}. \qedhere
    \]
\end{definition}

To reduce the number of states in the state graph to $\XP$ in $n$, we need to show that it suffices to search for schedules $S$ which satisfy $\cardinality{S_t} \le f(n)$ for all $t \in [\lifetime]$ for some function $f$.
We construct such schedules in the $\Sparsify$ procedure defined below.
The core idea of this procedure is to use the concentration property to delay the start of new temporal walks so that the size of $S_t$ is at most $f(n)$.
Specifically, let
\[
    f\colon n \mapsto 22 (n - 1).
\]
The following lemma is the core of the $\Sparsify$ operation, and the details of its proof are the reason for the specific choice of $f$.
For the rest of this section, the only important thing is that $f(n) \in \O(n)$.
To prove the lemma, we use the fact that every temporal walk starts its life by traversing a temporal edge, and as $G$ contains only $m = 2 \cdot (n - 1)$ edges, at most $m$ temporal walks can start their life at a given time step.
Note that for small values of $h$ (that is, for $h < f(n) / m$), this fact suffices to show that $\cardinality{S_t} \le f(n)$ because $f(n) / m$ time steps are needed to create $f(n)$ temporal walks, but by the time all walks are created, some of them are already more than $h$ time steps old.
For larger values of $h$, it may happen that $\cardinality{S_t} > f(n)$.
However, the following lemma shows that whenever this is the case, then there is some walk $p$ that is alive and idling on $c$ during the interval $[t-1, t+2]$.
Intuitively, as $p$ is idling, there is no point in introducing a new walk at this time step.
\begin{lemma}
    \label{thm:concentrated-walk-available}
    Let $S$ be a concentrated schedule and let $t$ be a minimal time step such that $\cardinality{S_t} > f(n)$.
    Then there exists a temporal walk $p \in S_t$ such that
    \begin{enumerate}[a)]
        \item $\firstt(p) \le t - 2$ and $\firstt(p) + h \ge t + 3$,
        \item the last temporal edge of $p$ before time step $t-1$ leads to the center $c$, and
        \item $p$ does not contain a temporal edge during $[t-1, t+2]$.\qedhere
    \end{enumerate}
\end{lemma}
\begin{proof}
    We count all possible ways that a walk $p \in S_t$ could violate one of the above conditions.
    If this count does not exceed $f(n)$, there must be at least one walk that satisfies all conditions.
    Recall that as $G$ contains $m = 2(n-1)$ edges and every temporal walk starts its life by traversing a temporal edge, at most $m$ temporal walks can start their life at any given time step.

    Let $p \in S_t$.
    By definition, $\firstt(p) \le t$ and $\firstt(p) \ge t - h$.
    To violate condition (a), we must have $\firstt(p) \ge t - 1$ or $\firstt(p) + h \le t + 2$, which eliminates up to $5m$ walks where $\firstt(p) \in [t-1, t] \cup [t-h, t-h+2]$.
    Otherwise, to violate condition (b), suppose that the last temporal edge of $p$ before $t-1$ leads to a leaf $v$.
    As $S$ is concentrated and $\firstt(p) \ge t+3-h$, we know that $p$ must return to $c$ within two time steps.
    For this to happen no earlier than time step $t-1$, the edge to $v$ must have been traversed at a time step in $[t-3, t-2]$, which eliminates $2m$ candidates.
    Finally, if $p$ fulfills conditions (a) and (b), to violate condition (c), an edge must be traversed in the interval $[t-1, t+2]$, for which there are only $4m$ temporal edges available.
    In total, this means that at most $11m$ temporal walks from $S_t$ violate at least one condition.
    As $\cardinality{S_t} > 11m$, there must be at least one $p \in S_t$ that fulfills all conditions.
\end{proof}

In the $\Sparsify$ procedure, we check whether there is a time step $t \in [\lifetime]$ at which $\cardinality{S_t}$ exceeds $f(n)$.
If $t$ is the earliest time step at which this is the case, then there must be a temporal walk $p \in S$ that traverses its first temporal edge at time step $t$.
By \cref{thm:concentrated-walk-available}, however, $p$ is not yet needed.
We delay the start of $p$ to decrease $\cardinality{S_t}$.
This process is repeated until $\cardinality{S_t} \le f(n)$ holds for all $t \in [\lifetime]$.
Notably, after modifying $p$, $S$ may not be concentrated anymore, so at the end of each iteration, we use the $\Concentrate$ procedure to restore concentration.

\begin{procedure}
    \SetupAlgoPrePost
    \AlgoPost{$S$ is concentrated}
    \AlgoPost{$\forall t \in [\lifetime]: \cardinality{S_t} \le f(n)$}

    call $\Concentrate(S)$ \;
    \While{$\exists t \in [\lifetime]: \cardinality{S_t} > f(n)$}{
        let $t \in [\lifetime]$ be minimal such that $\cardinality{S_t} > f(n)$ \;
        let $p \in S$ be a walk with $\firstt(p) = t$ \tcp*{exists because $\cardinality{S_{t-1}} < \cardinality{S_t}$}
        by \cref{thm:concentrated-walk-available} there is $q \in S$ waiting on $c$ during $[t-1, t+2]$ \;
        \If{the first temporal edge of $p$ is $(c, v, t)$}{
            \If{$\exists t' \in [t+1, t+2]: (v, c, t') \in \E(p)$}{
                remove $(c, v, t)$ and $(v, c, t')$ from $p$ \;
                add $(c, v, t)$ and $(v, c, t')$ to $q$ \; \label{line:sparsify-q-add-from-p}
            }
            \Else(\tcp*[h]{$S$ is concentrated, so $t + 2 \ge \lifetime + 1$}){
                remove $(c, v, t)$ from $p$ \;
                add $(c, v, t$) to $q$ \; \label{line:sparsify-q-add-late}
            }
        }
        \Else(\tcp*[h]{the first temporal edge of $p$ is $(v, c, t)$}){
            call \ShiftForwards($S$, $p$, $v$, $t-1$) \;
            remove $(v, c, t)$ from $p$ \;
            add $(c, v, t-1)$ and $(v, c, t)$ to $q$ \; \label{line:sparsify-q-add-allocated-and-p}
        }
        call $\Concentrate(S)$ \;
    }

    \caption{Sparsify($S$)}
\end{procedure}

\begin{lemma}
    \label{thm:sparsify-sound}
    Let $S$ be a feasible schedule and $p \in S$.
    All modifications made to $p$ in $\Sparsify$ ensure that $\E(p)$ still induces a connected sequence of temporal edges in $\G$.
\end{lemma}
\begin{proof}
    It is already shown in \cref{thm:concentrate-sound} that $\ShiftForwards$ maintains the desired properties.
    It remains to check all modifications made directly in $\Sparsify$.
    Let $p \in S$ be the temporal walk chosen in an iteration of $\Sparsify$ and $q \in S$ the temporal walk chosen using \cref{thm:concentrated-walk-available}.
    In all three cases of $\Sparsify$, the temporal edges that are removed from $p$ are the first temporal edges of $p$, so $p$ stays connected.
    The temporal edges added to $q$ are all contained in the time interval~$[t-1, t+2]$ during which $q$ does not contain any temporal edges before the addition.
    In Lines~\ref*{line:sparsify-q-add-from-p} and~\ref*{line:sparsify-q-add-allocated-and-p}, adding both temporal edges ensures that $q$ is on $c$ at time step~$t+3$, so $q$ is still connected to its next temporal edge afterward.
    In Line~\ref*{line:sparsify-q-add-late}, only a single temporal edge is added to $q$, so $q$ is on the vertex $v$ at time step $t+1$.
    However, in this case, we know that as $S$ is concentrated and $p$ does not contain a temporal edge during $[t+1, t+2]$, we must have $t + 2 \ge \lifetime + 1$.
    Therefore, $q$ does not contain any temporal edges at the time steps after $t+1$.
\end{proof}

\begin{lemma}
    \label{thm:sparsify-works}
    Let $S$ be a feasible, concentrated schedule and $S^*$ be the result of running one iteration of $\Sparsify$ on $S$.
    Then $\E(S) \subseteq \E(S^*)$, $S^*$ is TED, $\Phi(S) < \Phi(S^*)$, and for all $p^* \in S^*$ we have $\lifespan(p^*) \le h$.
\end{lemma}
\begin{proof}
    The proofs for all four statements are analogous to those of \cref{thm:concentrate-monotone,thm:concentrate-potential-inc,thm:concentrate-lifespan-constraint}.
    In particular, all statements are already shown for $\ShiftForwards$, so it only remains to check the modifications made directly in $\Sparsify$.
    When a temporal edge is removed from $S$, it is added back to another walk.
    Combined with \cref{thm:concentrate-monotone}, $\E(S) \subseteq \E(S^*)$ follows.
    Similarly, when a temporal edge is added to a walk, it was previously removed from another walk or freed using $\ShiftForwards$.
    Next, throughout $\Sparsify$, no walk is modified to start earlier than before, so the potential $\Phi(S)$ cannot decrease.
    However, each iteration delays the start of the chosen temporal walk $p$, thereby increasing the potential.
    Finally, when adding edges to the walk $q$ in each branch of $\Sparsify$, we know by choice of $q$ that doing so does not make $q$ violate the lifespan constraint.
\end{proof}

\begin{lemma}
    \label{thm:sparsify-terminates}
    $\Sparsify$ terminates.
\end{lemma}
\begin{proof}
    Analogous to \cref{thm:concentrate-terminates} using \cref{thm:sparsify-works}.
\end{proof}

Using \cref{thm:sparsify-works,thm:sparsify-terminates}, we know that $\Sparsify$ turns any feasible schedule into one with a bounded number of active walks at every time step.
Thus, when searching for solutions to a $\lifeTEDSC$ instance, we can limit the search space to just the states that are relevant for such schedules.
Before we present the algorithm, we make one more observation, allowing us to further reduce the number of states.
\begin{lemma}
    \label{thm:st-U-prefix}
    Let $\mathcal I$ be a $\lenTEDSC$ or $\lifeTEDSC$ instance and $S$ a feasible schedule for $\mathcal I$.
    For every path $P$ in $\St(\mathcal I)$ corresponding to $S$, there is a path $P'$ corresponding to the same schedule $S$, such that for all states $(t, U, A, p, c) \in V(P')$ the set $U$ is of the form $U = [k']$ for some $k' \in [0, k]$.
\end{lemma}
\begin{proof}
    Enumerate the temporal walks in $S$ as $\set{p_i}_{i = 1}^k$ such that they are in decreasing order by $\firstt(p_i)$.
    Thus, if a walk $p_j$ is still unused at time step $t$, then all walks $p_i$ with $i < j$ are still unused too.
    The path $P'$ is obtained by replacing every state of $P$ with the equivalent state in which the indices of the temporal walks are permuted in this order.
\end{proof}

\begin{theorem}
    There is an algorithm that solves $\lifeTEDSC$ on bidirected star graphs in $\cardinality{\mathcal I}^{\O(n)}$ time.
\end{theorem}
\begin{proof}
    Let $\mathcal I = (G, D, k, h)$ be a $\lifeTEDSC$ instance on a bidirected star.
    Set $\gapsize = (n - 1) \cdot k$.
    To decide whether $\mathcal I$ is feasible, we run the algorithm from \Cref{thm:life-len-k-xp}, but instead of constructing the compressed state graph $\St_\gapsize(\mathcal I)$ from \cref{def:compressed-state-graph} completely, we consider a subset of vertices.
    Let $f\colon \N \to \N$ with $f(n) \in \O(n)$ be defined as above and define $\St^\star(\mathcal I)$ as follows:
    The graph $\St^\star(\mathcal I) \subseteq \St_\gapsize(\mathcal I)$ should contain only those states $(t, U, A, p, c) \in V(\St_\gapsize(\mathcal I))$ satisfying $\cardinality{A} \le f(n)$ and $U = [k']$ for some $k' \in [0, k]$.
    Thus, the number of states $\cardinality{V(\St^\star(\mathcal I))}$ is bounded by
    \[
        \begin{bracealign}
            \underbrace{\cardinality D \cdot n \cdot k}_t
            \cdot
            \underbrace{(k + 1)}_U
            \cdot
            \underbrace{k^{f(n)}}_A
            \cdot
            \underbrace{n^{f(n)}}_p
            \cdot
            \underbrace{\cardinality{D}^{f(n)}}_c
        \end{bracealign} \le \cardinality{\mathcal I}^{\O(n)}.
    \]
    As $\St^\star(\mathcal I) \subseteq \St(\mathcal I)$, the proof of \Cref{thm:life-len-k-xp} already shows that if our algorithm outputs YES, then a feasible schedule exists.
    For the other direction, it suffices to show that the existence of any feasible schedule implies that a schedule exists which is witnessed by $\St^\star(\mathcal I)$.
    Suppose that $\mathcal I$ is feasible and let $S$ be a feasible schedule.
    After applying $\Sparsify$ to $S$, we can assume that $S$ is concentrated and satisfies $\cardinality{S_t} \le f(n)$ for all $t \in [\lifetime]$.
    By \cref{thm:sparsify-sound,thm:sparsify-works,thm:sparsify-terminates}, this modified schedule is still a feasible solution.
    We perform two modifications to $S$ to ensure that $S$ is witnessed in the compressed state graph $\St_\gapsize(\mathcal I)$.
    Both of these modifications only remove temporal edges, so for all $t \in [\lifetime]$ the size of $S_t$ does not increase.
    \begin{enumerate}
        \item
              Whenever a temporal walk $p \in S$ contains two successive temporal edges $(c, v, t^\vout), (v, c, t^\vin) \notin D$, then remove both of them from $p$.
              Let $[t_1, t_2 - 1]$ be a gap of size at least $\gapsize$.
              We show that, after this step, $p$ does not contain a temporal edge during $[t_1 + 2, t_2 - 3]$, unless it is the first or last temporal edge of $p$.
              Suppose that $p$ contains such a temporal edge.
              Then, combined with the previous or next temporal edge of $p$, $p$ contains two successive temporal edges of the form $(c, v, t^\vout)$ and $(v, c, t^\vin)$.
              As $S$ was concentrated before this step, we have $t^\vin \le t^\vout + 2$, so $t^\vout, t^\vin \in [t_1, t_2 - 1]$.
              Thus, both $t^\vout$ and $t^\vin$ are contained in the gap, so neither of the two temporal edges is contained in $D$, and they are both removed in this step.
        \item
              Trim the start and end of each temporal walk $p \in S$ such that the first and last temporal edge of $p$ is a demand in the draft schedule $D$.
              After this modification, for each gap $[t_1, t_2 - 1]$ of size at least $\gapsize$, no temporal walk moves during the interval $[t_1 + 2, t_2 - 1]$.
    \end{enumerate}
    Therefore, $S$ is witnessed by $\St_\gapsize(\mathcal I)$ and satisfies $\forall t \in [\lifetime]: \cardinality{S_t} \le f(n)$.
    This justifies the restriction of $\St^\star(\mathcal I)$ to only those states that satisfy $\cardinality{A} \le f(n)$.
    As a final step, we apply \cref{thm:st-U-prefix} to $S$, which justifies further restricting the states to only those with $U = [k']$ for some $k' \in [0, k]$.
    Thus, if a feasible schedule exists, then one exists that is witnessed by just the states of $\St^\star(\mathcal I)$.
\end{proof}

\section{Approximating Constrained TEDSC}\label{sec:approx}
So far we studied the schedule completion problem as a decision problem.
In this section, we switch perspective and investigate the corresponding search problem in which the output of an algorithm should be a schedule using the smallest possible number of walks.

\begin{problem}[]{\lang{Min-$\chi$-TEDSC} ($\chi \in \set{\len, \lifespan}$)}
    \Input & Static graph $G$, draft schedule $D$, integer $h$.            \\
    \Prob  & A schedule $S$ with $D \subseteq \E(S)$, $\chi(p) \le h$ for all $p \in S$, and of minimal size $\cardinality S$.
\end{problem}

We provide an approximation algorithm with regard to $\cardinality S$.
The main idea for the algorithm is to use the compressed static expansion $\SEcompressed{\gapsize}(G,D)$ from \cref{sec:unconstrained} and ask a slightly more involved question:
Instead of computing a flow of value $k$, we compute a \emph{minimum cost flow} of value $k$ for a suitable cost function.
This flow can then be transformed into a schedule of at most $2k - k/h$ walks.
The cost functions for $\lenTEDSC$ and $\lifeTEDSC$ are defined as follows.
\begin{definition}[$\len$-cost]
    The $\len$-cost of an edge $e = (u_t, v_{t'})$ is $c_\len(e) = d_G(u, v)$.
\end{definition}
\begin{definition}[$\lifespan$-cost]
    The $\lifespan$-cost of an edge $e = (u_t, v_{t'})$ is $c_\lifespan(e) = t' - t$.
\end{definition}
Observe that the cost of a path according to $c_\lifespan$ / $c_\len$ in the static expansion equals exactly the lifespan / length of the corresponding temporal walk in $\G$.
In particular, the $\ell$-cost of a waiting edge is 0; the $\ell$-cost of a moving edge in the uncompressed static expansion is 1; and the $\ell$-cost of a moving edge in the compressed static expansion is exactly the distance of the shortest static path between the corresponding vertices in $G$.
This motivates the following lemma.

\begin{lemma}
    \label{thm:approx-flow-exists}
    Let $\chi \in \set{\len, \lifespan}$ and let $S$ be a schedule for the \lang{$\chi$-TEDSC} instance $(G, D, k, h)$.
    Let $P_S$ be the set of static $s$-$z$-paths in $(\SEcompressed{(n-1)\cdot k}(G,D), r)$ from \Cref{def:compressed-tedsc-network} corresponding to $S$.
    Then $P_S$ induces a feasible flow $f_S$ of value $\cardinality S$ with
    \[
        c_\chi(f_S) = \sum_{p \in S} \chi(p)\le k \cdot h.
    \]
\end{lemma}

This means that if the minimum-cost flow has cost more than $k \cdot h$ then there is no feasible schedule.
In the other direction, a flow of cost at most $k \cdot h$ does not immediately yield a feasible schedule, since its total cost may be concentrated on only a few paths with cost larger than $h$.
To deal with this, we expand each compressed path to a path in the uncompressed static expansion $\SE(G,D)$ and split it greedily into segments of cost at most $h$, each of which becomes one walk in the schedule.

\begin{definition}
    For a static path $p$ in $\SE(G,D)$ we define its \emph{segment set} $\seg(p)$ as follows.
    If $c_\chi(p) = 0$ then $\seg(p) = \emptyset$.
    Else if $c_\chi(p) \le h$ then $\seg(p) = \set{p}$.
    Otherwise, let $p'$ be the longest prefix of $p$ that still satisfies $c_\chi(p') \le h$, and let $r$ be the path obtained by removing the edges of $p'$ from $p$.
    We set $\seg(p) = \set{p'} \cup \seg(r)$.
\end{definition}

In particular, every edge in $\SE(G,D)$ has $c_\chi$-cost at most $1$: lifespan-cost is exactly $1$ per edge, and length-cost is $0$ for waiting edges and $1$ for moving edges.
Therefore, every segment $p' \in \seg(p)$ except the last one satisfies $c_\chi(p') = h$.
Using this observation, we can bound the number of walks we generate by splitting.

\begin{lemma}
    \label{thm:approx-split-loss}
    Let $f$ be a feasible flow of value $k$ in $(\SEcompressed{(n-1)\cdot k}(G,D), r)$ such that $c_\chi(f) \le k \cdot h$.
    Let $P_f$ be a set of $k$ static $s$-$z$-paths inducing $f$ and $Q_f$ be the corresponding $k$ paths in $\SE(G,D)$ obtained by taking each walk in $P_f$ and (a) removing $s$ and $z$ and (b) expanding gaps according to \Cref{thm:tedw-n-times-k}.
    Consider the set $W_f$ of temporal walks in $\G(G,D)$ corresponding to $Q_f$.
    Then the set of segments
    \[
        W_f' = \bigcup_{q \in Q_f} \seg(q)
    \]
    viewed as temporal walks in $\G$, satisfies $\cardinality{W_f'} \le 2k - \frac k h$.
\end{lemma}
\begin{proof}
    We first expand $P_f$ to $Q_f$.
    For each $p_i \in P_f$ and each gap edge $(u, t_1) \to (v, t_2)$ in $p_i$, we use \Cref{thm:tedw-n-times-k}: path $i$ waits at $u$ until time $t_1 + (i-1)(n-1)$, traverses a shortest $u$-$v$-path in $G$, then waits at $v$ until $t_2$.
    This is valid since $t_2 - t_1 \ge (n-1)k$.
    Both $c_\len$ and $c_\lifespan$ are preserved by this expansion, so $\sum_{q \in Q_f} c_\chi(q) = c_\chi(f) \le k \cdot h$.
    Furthermore, the edges of the walks in $Q_f$ are distributed exactly to the walks in $W_f'$, except for the edges discarded in the case of $c_\chi(p) = 0$.
    However, any $p$ with $c_\chi(p) = 0$ traverses only waiting edges, so omitting such a segment from $W_f'$ does not alter the feasibility.

    The set $W_f'$ contains two kinds of segments, \emph{big} segments $p'$ with $c_\chi(p') = h$ and \emph{little} segments $p'$ with $c_\chi(p') < h$.
    We denote the sets of big and little segments as $B$ and $L$ respectively.
    The size of $B$ can be at most $k$ as otherwise $\sum_{p \in B} c_\chi(p) = \cardinality B \cdot h$ would exceed $k \cdot h$.
    Furthermore, the size of $L$ is at most $k$ because for each $q \in Q_f$ we have $\cardinality{L \cap \seg(q)} \le 1$.
    This already gives $\cardinality{W_f'} = \cardinality{B \cup L} \le 2k$.

    Observe that the two bounds for $\cardinality B$ and $\cardinality L$ are worst-case bounds.
    We can slightly improve the overall bound by analyzing the relation between these two terms.
    If $\cardinality B = k$ then $\cardinality L = 0$ because the cost of $B$ is already $k \cdot h$ and no walk in $L$ has cost $0$.
    Suppose that $\cardinality B = b$, then the remaining cost for the walks of $L$ is $(k - b) \cdot h$.
    As each segment costs at least $1$, this gives $\cardinality{B \cup L} \le b + \min\set{k, (k - b) \cdot h}$.
    This term is maximized at $b = k - \frac k h$ which gives $\cardinality{B \cup L} \le 2k - \frac k h$.
\end{proof}

\begin{algorithm}
    \SetupAlgoInputOutput
    \AlgoInput{\lang{Min-$\chi$-TEDSC} instance $(G, D, h)$}
    \AlgoOutput{Schedule $S$ with $\cardinality S \le (2 - h^{-1}) \OPT$}

    \For{$k = 0, 1, \dots, \cardinality D$}{
        construct $H = \SEcompressed{(n - 1)k}(G, D)$ and flow ranges $r$ according to \Cref{def:compressed-tedsc-network} \;
        \If{a feasible flow exists for $(H, r)$}{
            let $f$ be the minimum cost flow for $(H, r, c_\chi)$ \;
            \If{$c_\chi(f) \le k \cdot h$}{
                let $P$ be the set of $k$ paths in $\SEcompressed{(n-1)k}(G,D)$ inducing $f$ \;
                let $Q$ be the expansion of $P$ in $\SE(G,D)$ via \Cref{thm:tedw-n-times-k} \;
                \Return $S = \bigcup_{q \in Q} \seg(q)$ as temporal walks \;
            }
        }
    }

    \caption{\textsc{Min-$\chi$-TEDSC Approx.} ($\chi \in \set{\len, \lifespan}$)}
    \label{alg:two-approx}
\end{algorithm}

\begin{theorem}
    Given a \lang{Min-$\chi$-TEDSC} instance, \Cref{alg:two-approx}
    computes a feasible schedule $S$ with $\cardinality S \le (2 - h^{-1}) \OPT$ in polynomial time.
\end{theorem}
\begin{proof}
    \Cref{alg:two-approx} iterates over at most $\cardinality D + 1$ values of $k$, and each iteration runs a min-cost flow on $(\SEcompressed{(n-1) k}(G,D), r)$, which has polynomial size.
    The expansion and segmentation steps are also polynomial, so the algorithm runs in polynomial time.
    Note that a solution with $k = \cardinality D$ always exists, so let $S^*$ be a feasible schedule with $\cardinality{S^*} = \OPT \le \cardinality D$.
    By \Cref{thm:approx-flow-exists}, $S^*$ induces a feasible flow of value $\OPT$ and cost at most $\OPT \cdot h$ in $(\SEcompressed{(n-1)\cdot k}(G,D), r)$, so both conditions of the algorithm are satisfied at iteration $k = \OPT$.
    Since the algorithm returns at the first $k$ satisfying both conditions, it terminates at some $k^* \le \OPT$.
    By \Cref{thm:approx-split-loss} the returned schedule $S$ satisfies $\cardinality S \le 2k^* - \frac{k^*}{h} \le (2 - \frac 1 h) \OPT$.
\end{proof}

\section{Conclusion}
\label{sec:conclusion}
In this paper, we studied the complexity of $\TEDSC$, $\lenTEDSC$, and $\lifeTEDSC$.
We completely explored the parameterized complexity landscape with regard to the parameters $k$, $h$, and $\cardinality D$.
For restricted graph classes, we show that the complexities of the two variants differ.
One immediate open question is whether our polynomial-time algorithm on fixed bidirected stars can be generalized to trees.
Our proof heavily relies on the structure of star graphs, so we would need a different approach to generalize it.

To make $\TEDSC$ even more realistic, we can consider extending it in different ways.
First, we currently assume that the travel time of every edge is $1$, so one natural extension is to allow different edges to have varying travel times.
Our negative results extend to this setting, however, it is not clear whether our positive results would still hold.
The second extension is based on the start and end points of the walks.
In real-world transportation networks, only a subset of the stations can serve as depots where trains can start or end their trip.
We can extend $\TEDSC$ such that as part of the input we are given a subset of the vertices where the walks can start or end, and study whether the computationally hard variants become tractable.
Next, we current assume that any number of trains can reside at the same station concurrently.
This is an oversimplification of the real-world, which is rectified by specifying a capacity for each vertex in the input.
Finally, $\TEDSC$ considers time-stamped edge traversals as the elementary requirement.
To allow for more flexible objectives, we could allow for demands between vertices that are more than one edge away from each other.
In this context, it might also be sensible to specify that a trip should happen within a given time window rather than at a specified time step.

\section*{Acknowledgments}
The authors want to thank Ben Bals for useful discussions on \cref{sec:approx}.

\bibliography{references/references.bib}

\end{document}